\documentclass[a4paper, 11pt]{article}
\RequirePackage{fullpage}

\usepackage{amsthm,amsmath,amssymb}
\usepackage{subcaption}
\usepackage[usenames,dvipsnames]{xcolor}
\usepackage[pdftex,breaklinks,colorlinks,
    citecolor={BlueViolet}, linkcolor={Blue},urlcolor=Maroon]{hyperref}
\usepackage{tikz}
\usetikzlibrary{arrows.meta}
\usetikzlibrary{quotes}
\usetikzlibrary{calc}
\usepackage{charter,eulervm}%
\usepackage{enumerate}
\usepackage[final,expansion=alltext,protrusion=true]{microtype}

\usepackage{ctable}

\theoremstyle{plain} 
\newtheorem{theorem}{Theorem}[section]
\newtheorem{lemma}[theorem]{Lemma}
\newtheorem{corollary}[theorem]{Corollary}
\newtheorem{proposition}[theorem]{Proposition}
\newtheorem*{conjecture}{Conjecture}
\newtheorem*{definition}{Definition}

\newcommand{\FPT}{\textsf{FPT}}
\newcommand{\phcag}{proper Helly circular-arc graph}

\tikzstyle{filled vertex}  = [{circle,draw=blue,fill=black!50,inner sep=1pt}]  
\tikzstyle{empty vertex}  = [{circle, draw, fill = white, inner sep=1.5pt}]
\tikzstyle{uvertex} = [{violet, draw, fill=violet!50,inner sep=2pt}]  

\newcommand{\codecomment}[1]{\textbackslash\!\!\textbackslash {\;\em #1}}
\newcommand{\lp}[1]{\ensuremath{{\mathtt{lp}(#1)}}}
\newcommand{\rp}[1]{\ensuremath{{\mathtt{rp}(#1)}}}
\newcommand{\from}[1]{\ensuremath{{\mathtt{from}(#1)}}}
\newcommand{\stpath}[2]{$#1$--$#2$ path}
\newcommand{\stwalk}[2]{$#1$--$#2$ walk}

\title{Modification Problems \\toward Proper (Helly) Circular-arc Graphs}
\author{Yixin Cao\thanks{School of Computer Science and Engineering, Central South University, Changsha, China.}
  \thanks{Department of Computing, Hong Kong Polytechnic University, Hong Kong, China.  \texttt{yixin.cao@polyu.edu.hk}.}
    \and
    Jianxin Wang\footnotemark[1]
    \and
    Hanchun Yuan\footnotemark[1]
}
\date{}

\begin{document}
\maketitle

\begin{abstract}
  We present a $9^k\cdot n^{O(1)}$-time algorithm for the proper circular-arc vertex deletion problem, resolving an open problem of van ’t Hof  and Villanger [Algorithmica 2013] and Crespelle et al.~[arXiv:2001.06867].
  Our structural study also implies parameterized algorithms for modification problems toward proper Helly circular-arc graphs.
\end{abstract}

\section{Introduction} \label{sec:introduction}

A graph is a \emph{circular-arc graph} if its vertices can be assigned to arcs on a circle such that there is an edge between two vertices if and only if their corresponding arcs intersect.  If none of the arcs properly contains another, then the graph is a \emph{proper circular-arc graph}.
See Figure~\ref{fig:proper-cag} for two examples of proper circular-arc graphs.
Proper circular-arc graphs ``form an important subclass of the class of all claw-free graphs,'' and their study has been an important step towards finding ``a structural characterization of all claw-free graphs'' \cite{chudnovsky-08-claw-free-iii}.
The structures and recognition of proper circular-arc graphs have been well studied and well understood \cite{tucker-74-structures-cag, deng-96-proper-interval-and-cag}.

\begin{figure*}[h]
  \centering\small
  \begin{subfigure}[b]{0.45\linewidth}
    \centering
    \begin{tikzpicture}[scale=.7]
      \foreach[count =\j] \i in {1, 2, 3} 
      \draw ({120*\i-90}:1) -- ({120*\i+30}:1) -- ({120*\i-30}:2) -- ({120*\i-90}:1);
      \foreach[count =\j] \i in {1, 2, 3} {
        \node[empty vertex] (u\i) at ({210 - 120*\i}:2) {$u_{\i}$};
        \node[empty vertex] (v\i) at ({150 - 120*\i}:1) {$v_{\i}$};
      }
      \node at (-90:1.5) {};      
    \end{tikzpicture}
    \;
    \begin{tikzpicture}[scale=.11]
      \draw[dashed,thin] (10,0) arc (0:360:10);
      \foreach[count=\i] \from/\to/\radius in {100/-40/12, 340/200/13, 220/80/11} {
        \draw[thick] (\from:\radius) arc (\from:\to:\radius);
        \draw[olive, dashed] (\from:\radius) -- (\from:15);
        \node[olive, dashed] at (\from:16) {$v_{\i}$};
      }
      \foreach[count=\i] \from/\to/\radius in {120/60/13, 0/-60/11, 240/180/12} {
        \draw[thick] (\from:\radius) arc (\from:\to:\radius);
        \draw[olive, dashed] (\from:\radius) -- (\from:15);
        \node[olive, dashed] at (\from:16) {$u_{\i}$};
      }
    \end{tikzpicture}
    \caption{$S_3$}
  \end{subfigure}
  \begin{subfigure}[b]{0.45\linewidth}
    \centering
    \begin{tikzpicture}[scale=1.3]
      \node[empty vertex] (v5) at (0, 0) {$v_5$};
      \foreach \i in {1, ..., 4} {
        \node[empty vertex] (v\i) at ({135 - 90*\i}:1) {$v_{\i}$};
        \draw (v\i) -- (v5);
      }
      \draw (v1) -- (v2) -- (v3) -- (v4) -- (v1);
      \node at (-90:1.2) {};
    \end{tikzpicture}
    \;
    \begin{tikzpicture}[scale=.11]
      \draw[dashed,thin] (10,0) arc (0:360:10);
      \draw[thick] (180:13) arc (180.:360:13);
      \foreach[count=\i] \r in {11, 12, 11, 12} {
        \draw[thick] ({100-90*\i}:\r) arc ({100-90*\i}:{-10-90*\i}:\r);
        \draw[olive, dashed] ({100-90*\i}:\r) -- ({100-90*\i}:15);
        \node[olive, dashed] at ({100-90*\i}:16) {$v_{\i}$};
      }
      \draw[olive, dashed] (0:13) -- (0:15);
      \node[olive, dashed] at (0:16) {$v_{5}$};
    \end{tikzpicture}
    \caption{$W_4$}
  \end{subfigure}
  \caption{Two proper circular-arc graphs and their arc representations.}
  \label{fig:proper-cag}
\end{figure*}
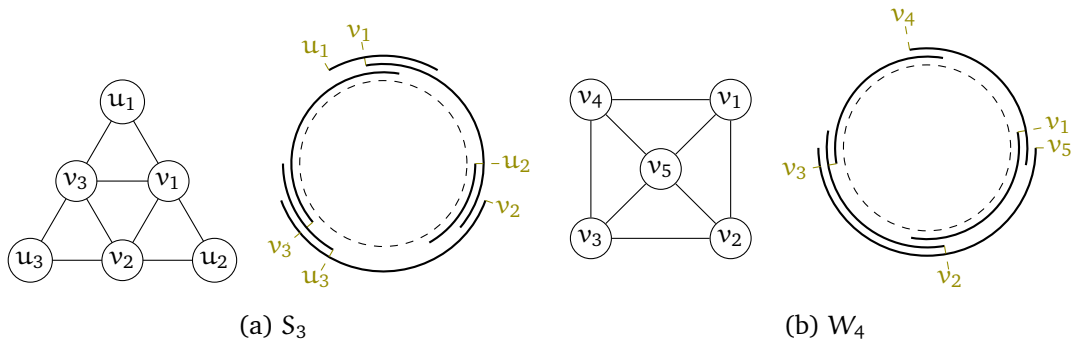

Another and earlier motivation of studying (proper) circular-arc graphs is from their relation with (proper) interval graphs, i.e., intersection graphs of intervals on the real line.  
It is easy to see that each (proper) interval graph is a (proper) circular-arc graph, and the connection of these classes has been used in both structural and algorithmic studies of these classes.
Indeed, the first  linear-time recognition algorithm for proper circular-arc graphs is based on a general observation of both proper circular-arc graphs and proper interval graphs~\cite{deng-96-proper-interval-and-cag}.  
Neither graph in Figure~\ref{fig:proper-cag} is a proper interval graph, but removing any vertex from Figure~\ref{fig:proper-cag}(a), or any vertex but $v_5$ from Figure~\ref{fig:proper-cag}(b) leaves a proper interval graph.

Let $\mathcal{G}$ be a graph class.  Given a graph $G$ and an integer $k$, the $\mathcal{G}$ vertex deletion problem asks whether we can remove $k$ vertices from $G$ to make a graph in $\mathcal{G}$.
These problems have been intensively studied in the framework of parameterized computation.  Suppose that the input graph has $n$ vertices and $m$ edges.  We say that a graph problem
is {\em fixed-parameter tractable (\FPT{})} if there is an algorithm solving it in time $f(k)\cdot n^{O(1)}$, where $f$ is a computable function depending only on~$k$ \cite{downey-13}.
For example, it is well known that the proper interval vertex deletion problem is \FPT{}\cite{villanger-13-pivd, cao-17-unit-interval-editing}.  In the algorithm of van 't Hof and Villanger~\cite{villanger-13-pivd}, the kinship between proper circular arc graphs and proper interval graphs plays a crucial role.  They show that it suffices to destroy all the small forbidden induced subgraphs, and then the graph is already a proper circular-arc graph, on which the proper interval vertex deletion problem can be solved in linear time.
They asked whether the proper circular arc vertex deletion problem is \FPT{} as well, and this open problem was recently raised again by Crespelle et al.~\cite{crespelle-20-survey-edge-modification}.  
We answer this question affirmatively.

\begin{theorem}\label{thm:alg-1}
  The proper circular-arc vertex deletion problem can be solved in time $9^k\cdot n^{O(1)}$.
\end{theorem}

A major difference between the class of proper interval graphs and the class of proper circular-arc graphs is that the later class is not closed under disjoint union.  This can be easily observed from their representations: while we can always put intervals for two different components side by side, no such accommodation is possible for two set of arcs if one set of them covers the whole circle.  As a matter of fact, if a proper circular-arc graph is not connected, it has to be a proper interval graph.
  (The same remark applies to the relation between circular-arc graphs and interval graphs.)

  If a proper circular-arc graph contains a hole of length at least five, then its property is quite similar to a proper interval graph.  What is difficult is when a few arcs cover the whole circle in an arc representation.  For such a graph, it is more convenient to study its complement.  Indeed, when characterizing proper circular-arc graphs, Tucker~\cite{tucker-74-structures-cag} actually listed the forbidden induced subgraphs of the complement class.
  He also observed that if the complement  $\overline{G}$ of a proper circular-arc graph $G$ is not connected, then $\overline{G}$ is bipartite.
\emph{Permutation graphs} are the intersection graphs of line segments between two parallel lines, and \emph{bipartite permutation graphs} are those permutation graphs that are bipartite.
Bipartite permutation graphs are also known as {proper interval bigraphs} and unit interval bigraphs \cite{hell-04-interval-bigraphs-and-cag}.
It is well known that a co-bipartite graph $H$ is a proper circular-arc graph if and only if $\overline{H}$ is bipartite permutation graph. 

Let $(G, k)$ be an instance to the proper circular-arc vertex deletion problem, and let $V_-$ be a solution.  If $G - V_-$ is not connected, then it is a proper interval graph; if ${G - V_-}$ is not co-connected, then it is the complement of a bipartite permutation graph.  We can call the algorithm of Cao~\cite{cao-17-unit-interval-editing} and the algorithm of Bo\.{z}yk et al.~\cite{bozyk-20-bipartite-permutation} to check whether such a set $V_-$ exists, and we are done if it does.
In the rest we may assume that $G - V_-$ is neither a proper interval graph nor the complement of a bipartite graph, hence both connected and co-connected.  For this purpose we may assume that $G$ itself is both connected and co-connected; otherwise, there is a unique component $C$ of $G$ or $\overline G$ such that $V(G)\setminus V(C)\subseteq V_-$.  Either the instance is trivially \FPT{}, when $n = O(k)$, or it suffices to consider the largest component of $G$ or $\overline G$.

The algorithm proceeds as follows.  We can destroy all forbidden induced subgraphs of order at most seven by branching.  Now $G$ is free of small forbidden induced subgraphs and is both connected and co-connected.
Our key observation is that if ${G}$ is not already a proper circular-arc graph, then $\overline{G}$ must be bipartite.  Note that any induced subgraph of a bipartite graph is bipartite, but we have assumed that $G - V_-$ is not the complement of a bipartite graph.
Therefore, we are ready to directly return ``yes'' or ``no.''

Since the parameterized algorithm branches on a small set of vertices that intersects every solution, we can easily turn it into an approximation algorithm for the maximum proper circular-arc induced subgraph problem.

\begin{theorem}\label{thm:alg-2}
  There is a polynomial-time approximation algorithm of approximation ratio~$9$ for the minimization version of the proper circular-arc vertex deletion problem.
\end{theorem}

Proper circular-arc graphs also arise naturally when we consider the clique graph (the intersection graph of maximal cliques of a host graph) of a circular-arc graph.
The complicated structures of circular-arc graphs are mainly due to the lack of the so-called \emph{Helly} property: every set of pairwise intersecting arcs has a common intersection.
For example, neither representation in Figure~\ref{fig:proper-cag} is Helly: the set $\{v_1, v_2, v_3\}$ in (a) and the set $\{v_3, v_4, v_5\}$ in (b) violate the Helly property.
A graph is a \emph{Helly circular-arc graph} if it admits an arc representation that is Helly.
Since every interval representation is Helly,  all interval graphs are Helly circular-arc graphs.  It is well known that the clique graph of an interval graph, with at most $n$ maximal cliques, is a proper interval graph \cite{hedman-84-clique-graphs}.  The same upper bound holds for the number of maximal cliques in a Helly circular-arc graph, and the clique graph of a Helly circular-arc graph is always a proper circular-arc graph \cite{duran-01-clique-graphs-helly-cag}.
Let us remark parenthetically that a circular-arc graph may have an exponential number of maximal cliques, e.g., the complement of the union of $p$ disjoint edges, which has $2 p$ vertices, each of degree $2 p - 2$.

A graph is a proper Helly circular-arc graph if and only if its clique matrix has the circular-ones property for both rows and columns \cite{lin-13-nhcag-and-subclasses}.

The class of proper Helly circular-arc graphs is sandwiched between proper circular-arc graphs and proper interval graphs.  This observation has been crucial for our algorithms for modification problems toward proper interval graphs \cite{cao-17-unit-interval-editing}.
A graph is a \emph{proper Helly circular-arc graph} if it has an {arc representation} that is both proper and Helly.
A word of caution is worth on the definition of \phcag{s}.  One graph might admit two arc representations, one being proper and the other Helly, but no arc representation that is both proper and Helly.
For example, both representations in Figure~\ref{fig:proper-cag} are proper but neither is Helly, and it is not difficult to make Helly arc representations for $S_3$ and $W_4$, but, as the reader may easily verify, neither of them admits an arc representation that is both proper and Helly.  
Therefore, the class of proper Helly circular-arc graphs does not contain all those graphs being both proper circular-arc graphs and Helly circular-arc graphs, but a proper subclass of it.
  Indeed, a proper circular-arc graph is a proper Helly circular-arc graph if and only if it is $\{S_3, W_4\}$-free \cite{lin-07-phcag}.

We then consider modification problems toward proper Helly circular-arc graphs.  For this class we consider also the edge deletion and completion problems (a proof of their NP-completeness was provided in the appendix).
Again, we start by destroying all small forbidden induced subgraphs, up to six vertices.  We show that a connected graph free of such induced subgraphs is already a proper Helly circular-arc graph.
For the vertex deletion problem, either we remove all but one component, or we remove vertices to get a proper interval graph.  The edge deletion problem is even simpler: if the graph is not connected, we cannot make it connected by deleting edges.  Thus, depending on whether the graph is connected, either we are already done, or we are solving the proper interval edge deletion problem.  This idea can even solve the general deletion problem that allows $k_1$ vertex deletions and $k_2$ edge deletions.
The situation is quite different for the completion problem.  We are happy if we can add at most $k$ edges to make the input graph a proper interval graph.  Otherwise, we have to make a connected proper Helly circular-arc graph.  After we have dealt with all the small forbidden induced subgraphs, the only nontrivial case is when there is a large component, which contains a long hole $H$, and several small components.  We need to ``attach'' these small components to vertices on $H$.  Since these operations are local, we can find a solution by dynamic programming.
Thus, all the three problems are \FPT{}, and they can be done in linear \FPT{} time.
Again, the parameterized algorithm for the vertex deletion problem can be easily turned into an approximation algorithm.
\begin{theorem}\label{thm: main-theorem2}
  For modification problems toward proper Helly circular-arc graphs, there are
  \begin{itemize}
  \item an $O(6^k\cdot (m + n))$-time algorithm for the vertex deletion problem;
  \item an $O(8^k\cdot (m + n))$-time algorithm for the edge deletion problem;
  \item an $O(14^{k_1+k_2}\cdot (m + n))$-time algorithm for the deletion problem; and
  \item a $k^{O(k)}\cdot (m + n)$-time algorithm for the completion problem.
  \end{itemize}
  Moreover, there is an $O(n m + n^2)$-time approximation algorithm of approximation ratio~$6$ for the minimization version of the proper Helly circular-arc vertex deletion problem.
\end{theorem}

Somewhat surprisingly, modification problems toward circular-arc graphs and its subclasses have not received sufficient attention.
We hope our work will inspire more study in this direction.
Apart from the two classes in the present paper, the next interesting class is the class of normal Helly circular-arc graphs, a super class of proper Helly circular-arc graphs.
They have played crucial roles in solving modification problems to interval graphs \cite{cao-15-interval-deletion, cao-16-almost-interval-recognition}.
Also related and probably simpler are the modification problems toward unit (Helly) circular-arc graphs.  It is well known that a graph is a proper interval graph if and only if it is a unit interval graph.
However, there are proper (Helly) circular-arc graphs that are not unit (Helly) circular-arc graphs, e.g., the graph obtained from an even hole of length at least eight by adding edges to connect consecutive even-numbered vertices.  
We refer interested readers to Lin et al.~\cite{lin-13-nhcag-and-subclasses} and Cao~\cite{cao-17-unit-interval-editing} for the hierarchy of subclasses of circular-arc graphs.

\section{Preliminaries}\label{sec:structures}

All graphs discussed in this paper are undirected and simple.  The vertex set and edge set of a graph $G$ is denoted by, respectively, $V(G)$ and $E(G)$.  Let $n = |V(G)|$ and $m = |E(G)|$.
A \emph{walk}\index{walk} in a graph $G$ is a sequence of vertices and edges in the form of $v_0$, $v_0 v_1$, $v_1$, $v_1 v_2$, $\ldots$, $v_\ell$.  Since the edges are determined by the vertices, such a walk can be denoted by $v_0 v_1 \ldots v_\ell$ unambiguously.  We say that this walk connects $v_0$ and $v_\ell$, which are the \emph{ends}\index{end!of a walk} of this walk, and refer to it as a \stwalk{v_0}{v_\ell}\index{walk!$s$--$t$ walk}.  The \emph{length} of a walk is the number of occurrences of edges it contains, and $\ell$ in the previous example.  A walk is \emph{closed} if $\ell > 1$ and $v_0 =v_\ell$.  A walk is a \emph{path} if all its vertices are distinct.
A closed walk of length $\ell$ is a \emph{cycle}\index{cycle} if it visits precisely $\ell$ vertices; i.e., no repeated vertices except the two ends.  A \emph{hole} is an induced cycle of length at least four.  A walk, path, cycle, or hole is \emph{odd} (resp., \emph{even}) if its length is odd (resp., even).
For $\ell\ge 3$, we use $C_\ell$ to denote an induced cycle on $\ell$ vertices; if we add a new vertex to a $C_\ell$ and make it adjacent to no or all vertices on the cycle, then we end with a $C_\ell^*$ or $W_\ell$, respectively.

The complement graph $\overline{G}$ of a graph $G$ is defined on the same vertex set $V(G)$, where a pair of vertices $u$ and $v$ is adjacent in $\overline{G}$ if and only if $u v \not\in E(G)$; e.g., $\overline{C_{5}^*}$ is ${W_5}$.  The graph $\overline{C_{3}^*}$ is also called a \emph{claw}.
A graph $G$ is \emph{connected} if every pair of vertices is connected by a path, and \emph{co-connected} if $\overline{G}$ is connected.

A \emph{circular-arc graph} is the intersection graph of a set of arcs on a circle.  The set of arcs is called an \emph{arc representation} of this graph.  
In this paper, all arcs are closed.  An arc representation is \emph{proper} if no arc in it properly contains another arc.  A graph is a \emph{proper circular-arc graph} if it has a proper arc presentation.
In case that there is a point of the cycle avoided by all the arcs in an arc representation, we can cut the circle and straighten all the arcs into line segments.
Such a graph is an \emph{interval graph}, i.e., the intersection graph of a set of closed intervals on the real line, and the set of intervals is an \emph{interval representation} of this graph.
Proper interval graphs are defined analogously.
Clearly, any (proper) interval representation can be viewed as a (proper) arc representation leaving some point uncovered, and hence all (proper) interval graphs are always (proper) circular-arc graphs.

\begin{figure}[h]
  \centering\small
  \begin{subfigure}[b]{0.15\linewidth}
    \centering
    \begin{tikzpicture}[every node/.style={filled vertex}, scale=.44]
      \foreach[count =\j] \i in {1, 2, 3} 
        \draw ({120*\i-90}:1) -- ({120*\i+30}:1) -- ({120*\i-30}:2) -- ({120*\i-90}:1);
        \foreach[count =\j] \i in {1, 2, 3} {
          \node (u\i) at ({120*\i-30}:2) {};
          \node (v\i) at ({120*\i-90}:1) {};
      }
    \end{tikzpicture}
    \caption{$S_{3}$ (tent)}\label{fig:tent}
  \end{subfigure}
  \,
  \begin{subfigure}[b]{0.16\linewidth}
    \centering
    \begin{tikzpicture}[every node/.style={filled vertex},scale=.22]
      \node (v1) at (0,6) {};
      \node (u1) at (0,3.5) {};
      \node (v2) at (-5,0) {};
      \node (u2) at (-2,0) {};
      \node (u3) at (2,0) {};
      \node (v3) at (5,0) {};
      \draw[] (v2) -- (u2) -- (u3) -- (v3);
      \draw[] (u1) -- (v1);
      \draw[] (u2) -- (u1) -- (u3);
    \end{tikzpicture}
    \caption{$\overline{S_{3}}$ (net)}\label{fig:net}    
  \end{subfigure}
   \,
  \begin{subfigure}[b]{.18\linewidth}
    \centering
   \begin{tikzpicture}[scale=.6]
    \node [filled vertex] (a) at (-2, 0) {};
    \node [filled vertex] (b) at (-1, 0) {};
    \node [filled vertex] (c) at (0,0) {};
    \node [filled vertex] (d) at (1,0) {};
    \node [filled vertex] (e) at (2,0) {};
    \node [filled vertex] (v) at (0, 1.15) {};
    \node [filled vertex] (u) at (0, 2.3) {};
    \draw (a) -- (b) -- (c) -- (d) -- (e);
    \draw (u) -- (v) -- (c) ;
    \end{tikzpicture}
    \caption{$F_{1}$ (long claw)}
  \end{subfigure}  
  \,
  \begin{subfigure}[b]{0.11\linewidth}
    \centering
    \begin{tikzpicture}[every node/.style={filled vertex}, scale=.7]
      \node (v1) at (0,1) {};
      \node (v2) at (0,0) {};
      \node (v3) at (1.5,0) {};
      \node (v4) at (1.5,1) {};
      \node (v5) at (1.5,2) {};
      \node (v6) at (0,2) {};
      \node (v7) at (0.75,0.5) {};
      \draw (v1) -- (v2)  -- (v3) -- (v4) -- (v5) (v6) -- (v1) -- (v4);
      \draw (v7) edge (v3);
    \end{tikzpicture}
    \caption{$F_{2}$}
  \end{subfigure}  
  \,
  \begin{subfigure}[b]{0.11\linewidth}
    \centering
    \begin{tikzpicture}[every node/.style={filled vertex}, scale=.7]
       \node (v1) at (0,1) {};
      \node (v2) at (0,0) {};
      \node (v3) at (1.5,0) {};
      \node (v4) at (1.5,1) {};
      \node (v5) at (1.5,2) {};
      \node (v6) at (0,2) {};
      \node (v7) at (0.75,0.5) {};
      \draw (v1) -- (v2)  -- (v3) -- (v4) -- (v5) -- (v6) -- (v1) -- (v4);
      \draw (v7) edge (v1);
    \end{tikzpicture}
    \caption{$F_{3}$}
  \end{subfigure}  
  \,
  \begin{subfigure}[b]{0.11\linewidth}
    \centering
    \begin{tikzpicture}[every node/.style={filled vertex}, scale=.7]
      \node (v1) at (0,1) {};
      \node (v2) at (0,0) {};
      \node (v3) at (1.5,0) {};
      \node (v4) at (1.5,1) {};
      \node (v5) at (1.5,2) {};
      \node (v6) at (0,2) {};
      \node (v7) at (0.75,0.5) {};
      \draw (v1) -- (v2)  -- (v3) -- (v4) -- (v5) -- (v6) -- (v1) -- (v4);
      \draw  (v1) -- (v7) -- (v4);
    \end{tikzpicture}
    \caption{$F_{4}$}
  \end{subfigure}
  \caption{Some small forbidden induced graphs.}
  \label{fig:small-graphs}
\end{figure}
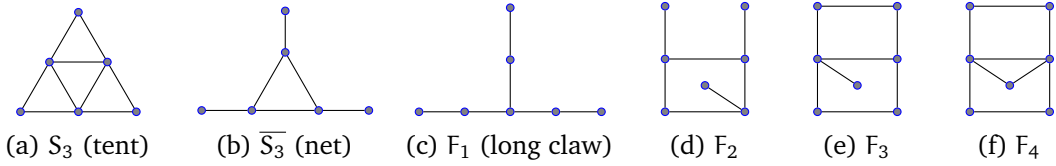

Let $F$ be a fixed graph.  We say that a graph $G$ is \emph{$F$-free} if $G$ does not contain $F$ as an induced subgraph.  For a set $\mathcal{F}$ of graphs, a graph $G$ is \emph{$\mathcal{F}$-free} if $G$ is $F$-free for every $F\in \mathcal{F}$.  If every $F\in \mathcal{F}$ is minimal, i.e., not containing any $F'\in \mathcal{F}$ as a proper induced subgraph, then $\mathcal{F}$  %
are the (minimal) \emph{forbidden induced subgraphs} of this class.  See Figure~\ref{fig:small-graphs} for some of the forbidden induced subgraphs considered in the present paper.
We use $S_3^*$ to denote the graph obtained by adding an isolated vertex to $S_3$.

\begin{theorem}[\cite{tucker-74-structures-cag}]
  \label{thm:proper-forbidden-induced-subgraphs}
  A graph is a proper circular-arc graph if and only if it is free of $S_3^*$, $C_\ell^*$ with $\ell \ge 4$, as well as the complement of $S_3$, $F_{1}$, $F_{2}$, $F_{3}$, $F_{4}$, ${C_{2 \ell + 2}}$, and ${C_{2 \ell - 1}^*}$ with $\ell \ge 2$.
\end{theorem}

Neither the class of circular-arc graphs nor the class of proper circular-arc graphs is closed under taking disjoint unions.
Indeed, 
if a (proper) circular-arc graph $G$ is not a (proper) interval graph, then in any representation of $G$, the union of the arcs covers the whole circle.  Such a graph is necessarily connected.

\begin{proposition}[Folklore]
  \label{lem:connected}
  If a proper circular-arc graph $G$ is not connected, then $G$ is a proper interval graph.
\end{proposition}

Proper circular-arc graphs have three infinite families of forbidden induced subgraphs, namely, $\{C_\ell^*\mid \ell \ge 4\}$, $\{\overline{C_{2 \ell + 2}}\mid \ell \ge 2\}$, and $\{\overline{C_{2 \ell - 1}^*}\mid \ell \ge 2\}$.  The first of them can be ignored for connected graphs.
\begin{lemma}\label{lem:small-subgraphs-free}
  Let $G$ be a connected graph.  If $G$ does not contain the complement of $C_3^*$ or $S_3$,
  then $G$ is $\{C_\ell^* \mid \ell\ge 5\}$-free.
\end{lemma}
\begin{proof}
  Suppose for contradiction that there exist an induced cycle $C$ and a vertex $v$ in $G$ with $|C| \ge 5$ and $V(C)\cap N(v) = \emptyset$.
  Since $G$ is connected, we can find a shortest path from $v$ to $C$.  Let the last three vertices on this path be $x$, $y$, and $z$; note that $z$ is on $C$ and $x$ is nonadjacent to any vertex on $C$.
  We may number the vertices on $C$ such that $C = v_{1} v_{2}\cdots v_{|C|}$ and $z = v_{2}$.
  If $y$ is adjacent to only $v_{2}$ on $C$, then $\{v_{1}, v_{2},v_{3}, y\}$ induces a claw.  If $y$ is also adjacent to both $v_{1}$ and $v_{3}$, then $\{v_{1},v_{3}, x, y\}$ induces a claw, and it is similar if $y$ is adjacent to any three consecutive vertices on $C$.  Otherwise, $y$ is adjacent to precisely one of $v_{1}$ and $v_{3}$.  Without loss of generality, assume that $y$ is adjacent to $v_{3}$ but not $v_{1}$.  Note that $y$ is not adjacent to $v_{4}$ either, and then $\{v_{1}, v_{2}, v_{3}, v_{4}, x, y\}$ induces a copy of the complement of $S_{3}$.
\end{proof}

An arc representation is {\em Helly} if every set of pairwise intersecting arcs has a common intersection.  A {circular-arc graph} is \emph{proper Helly} if it has an {arc representation} that is both proper and Helly.

\begin{theorem}[\cite{lin-13-nhcag-and-subclasses}]
  \label{thm:pcag-phcag}
  A proper circular-arc graph is a \phcag{} if and only if it contains no $W_4$ or $S_3$.
\end{theorem}

Note that $S_3^*$ contains $S_3$, while all the complements of $F_1$, $F_{2}$, $F_{3}$, $F_{4}$, and $\{ {C_{2 \ell}}, {C_{2 \ell - 1}^*} \mid \ell \ge 4\}$ contain $W_{4}$.  The following corollaries follow from Theorem~\ref{thm:pcag-phcag}, together with Theorem~\ref{thm:proper-forbidden-induced-subgraphs} and Lemma~\ref{lem:small-subgraphs-free}, respectively.

\begin{corollary}[\cite{lin-13-nhcag-and-subclasses}]
  \label{thm:phcag}
  A graph is a \phcag{} if and only if it contains no $\overline{C^*_3}$, $S_3, \overline{S_3}$, $W_4$, $W_5$, $\overline{C_6}$, or $C_\ell^*$ for $\ell \ge 4$.
\end{corollary}

\begin{corollary}\label{cor:phcag-small-fis}
  Let $G$ be a connected graph.
  If $G$ does not contain
  $\overline{C_3^*}$, $C_4^*$, $S_3, \overline{S_3}$, $W_4$, $W_5$, or $\overline{C_6}$, then $G$ is a proper Helly circular-arc graph.
\end{corollary}

Recall that proper interval graphs are precisely $\{\overline{C_3^*}, S_3, \overline{S_3}, C_\ell\mid \ell \ge 4\}$-free graphs.
\begin{corollary}\label{cor:phcag-uig}
  Let $G$ be a proper Helly circular-arc graph.
  Then $G$ is a proper interval graph if and only if $G$ does not contain any holes.
\end{corollary}

The following can be viewed as a constructive version of Corollary~\ref{cor:phcag-small-fis}. 

\begin{proposition}\label{prop:find-small-fis}
  Let $G$ be a connected graph.
  In $O(m + n)$ time we can either detect %
  an induced subgraph in $\{ \overline{C_3^*}, C_4^*, S_3, \overline{S_3}, W_4, W_5, \overline{C_6}\}$, or build a proper and Helly arc representation for $G$.
\end{proposition}
\begin{proof}
  We start by calling the certifying recognition algorithm for \phcag{s} \cite{cao-17-nhcag}.
  If a proper and Helly arc representation is found, then we return it.  In the rest, the algorithm returns a forbidden induced subgraph.  If the forbidden induced subgraph is one of 
  $\overline{C_3^*}, C_4^*, S_3, \overline{S_3}, W_4, W_5$, and $\overline{C_6}$, then we return its vertex set.  Otherwise, the outcome must be a $C_{\ell}^*$, where $\ell\ge 5$, and let $H = v_{1} \cdots v_{\ell}$ be the hole and $v$ the isolated vertex.  We can use breadth-first search to find a shortest path from $v$ to $H$.  Let $x$ and $y$ be the last two vertices on this path that are not on $H$; note that $y$ is adjacent to $H$ while $x$ is not.  If $y$ has only one neighbor $v_i$ on $H$, then we return $\{v_i, y, v_{i-1}, v_{i+1}\}$ as a claw.  If $y$ has exactly two consecutive neighbors on $H$, say, $v_{i}$ and $v_{i+1}$, then we return $\{x, y, v_{i-1}, v_{i}, v_{i+1}, v_{i+2}\}$ as an $\overline{S_3}$.  In the rest, $y$ is adjacent to two vertices $v_{i}$ and $v_{j}$ on $H$ with $|i - j| > 1$, then we return $\{y, x, v_i, v_{j}\}$ as a claw.
\end{proof}

A graph is a \emph{permutation graph} if its vertices can be assigned to line segments between two parallel lines such that there is an edge between two vertices if and only if their corresponding segments intersect.
The class of permutation graphs has a large number of forbidden induced subgraphs \cite{gallai-67-transitive-orientation}.  Fortunately, most of them contain an odd cycle, and thus the structures of forbidden induced subgraphs of bipartite permutation graphs are far simpler.
\begin{theorem}[\cite{gallai-67-transitive-orientation}]
  \label{thm:bipartite-permutation}
  A graph is a bipartite permutation graph if and only if it is free of $F_{1}$, $F_{2}$, $F_{3}$, ${C_{2 \ell + 2}}$, and ${C_{2 \ell - 1}}$ with $\ell \ge 2$.
\end{theorem}

The following statement correlates proper circular-arc graphs and bipartite permutation graphs.
\begin{theorem}[Folklore]\label{thm:co-bipartitie}
  The following are equivalent on a graph $G$:
  \begin{enumerate}[i)]
  \item $G$ is a proper circular-arc graph and $\overline{G}$ is bipartite; and
  \item $\overline{G}$ is a bipartite permutation graph.
  \end{enumerate}
\end{theorem}
\begin{proof}
  i) $\rightarrow$ ii).  By Theorem~\ref{thm:proper-forbidden-induced-subgraphs}, $\overline{G}$ is free of $F_{1}$, $F_{2}$, $F_{3}$, and ${C_{2 \ell + 2}}, \ell \ge 2$.  On the other hand, since $\overline{G}$ is bipartite, it does not contain any odd cycle,  Thus, the claim then follows from Theorem~\ref{thm:bipartite-permutation}.

  ii) $\rightarrow$ i).
  By Theorem~\ref{thm:bipartite-permutation}, $\overline{G}$ is free of $F_{1}$, $F_{2}$, $F_{3}$, ${C_{2 \ell + 2}}$, and ${C_{2 \ell - 1}}$ with $\ell \ge 2$.  Moreover, $\overline{G}$ cannot contain $S_3$, $\overline{S_3^*}$, $F_{4}$, or $\overline{C_\ell^*}$ with $\ell \ge 4$ because each of them contains a triangle; $\overline{G}$ is free of ${C_{2 \ell - 1}^*}$ with $\ell \ge 2$ because ${C_{2 \ell - 1}^*}$ contains an odd cycle.
\end{proof}

The following is complement to Proposition~\ref{lem:connected} in a sense.  Note that (proper) circular-arc graphs that are co-bipartite have played crucial roles in understanding these graph classes \cite{tucker-74-structures-cag}.
\begin{proposition}\label{lem:co-bipartitie}
  Let $G$ be a proper circular-arc graph.  If $\overline{G}$ is not connected,  then $\overline{G}$ is a bipartite permutation graph.
\end{proposition}
\begin{proof}
  Since ${G}$ is a proper circular-arc graph, $\overline{G}$ does not contain $C_{2\ell - 1}^*$ or $C_{2\ell + 2}$ with $\ell\geq 2$.  If $\overline{G}$ is not connected and contain any odd cycle, then this cycle and any vertex from another component form a $C_{2\ell - 1}^*$.  Therefore, $\overline{G}$ is bipartite, and the statement follows from Theorem~\ref{thm:co-bipartitie}.
\end{proof}

\section{Deletions to proper Helly circular-arc graphs}

We first study the proper Helly circular-arc vertex deletion problem.
We may assume without loss of generality that the input graph cannot be made a proper interval graph by removing $k$ vertices.
Therefore, the resulting graph after removing any $k$-solution is connected by Proposition~\ref{lem:connected}.
An \FPT{} algorithm is immediate from Corollary~\ref{cor:phcag-small-fis}: after destroying all the copies of $\overline{C_3^*}$, $C_4^*$, $S_3, \overline{S_3}$, $W_4$, $W_5$, and $\overline{C_6}$ in $G$ by standard branching, we return all vertices except those in a maximum-order component.  A similar (and simpler) approach works for the proper Helly circular-arc edge deletion problem.
The focus of the following proof is thus on efficient implementations.

\begin{theorem}\label{thm:algs-phcag-deletion}
  The proper Helly circular-arc vertex deletion problem and the proper Helly circular-arc edge deletion problem can be solved in time $O(6^k\cdot (m + n))$ and $O(10^k\cdot (m + n))$, respectively.
\end{theorem}
\begin{proof}
  Let ($G, k$) be an instance of proper Helly circular-arc vertex deletion.  Our algorithm proceeds as follows.
  We start by calling the algorithm of Cao~\cite{cao-17-unit-interval-editing} to check whether there is a set $V_-$ of at most $k$ vertices such that $G - V_-$ is a proper interval graph.  If the set is found, then we return ``yes.''  In the rest, we look for a solution $V_-$ such that $G - V_-$ is \textit{not} a proper interval graph.  By Proposition~\ref{lem:connected}, (note that a proper Helly circular-arc graph is a proper circular-arc graph,) $G - V_-$ is connected.

  For the general case, the algorithm solves the problem by making recursive calls to itself; we return ``no'' directly for a recursive call in which $k < 0$.
  For each component $C$ of $G$, we call the algorithm of Proposition~\ref{prop:find-small-fis}.  If an induced subgraph $G[F]$ is found, then the algorithm calls itself $|F|$ times, each with a new instance $(G - v, k - 1)$ for some $v\in F$.  Since we need to delete at least one vertex from $F$, the original instance $(G, k)$ is a yes-instance if and only if at least one of the instances $(G - v, k - 1)$ is a yes-instance.  Now that $G$ is free of $\overline{C_3^*}$, $C_4^*$, $S_3, \overline{S_3}$, $W_4$, $W_5$, and $\overline{C_6}$, every component of $G$ is a proper Helly circular-arc subgraph (Corollary~\ref{cor:phcag-small-fis}).  We find a component $C$ of $G$ that has the maximum order.  We return ``yes'' if $|V(C)| \ge n - k$, when $V(G)\setminus V(C)$ is a solution, or ``no'' otherwise.
    Since each of $\overline{C_3^*}, C_4^*, S_3, \overline{S_3}, W_4, W_5$, and $\overline{C_6}$ has at most $6$ vertices, at most $6$ recursive calls are made, all with parameter value $k - 1$.  By Proposition~\ref{prop:find-small-fis}, each recursive call can be made in $O(m + n)$ time.  Therefore, the total running time is $O(6^k\cdot (m + n))$.

    The algorithm for the edge deletion problem is even simpler.  Again, we start by calling the algorithm for proper interval edge deletion problem~\cite{cao-17-unit-interval-editing}, which takes time $O(4^{k} \cdot (m + n))$.  We proceed only when the answer is ``no.''
    In the recursive calls for the general case, we always return ``no'' whenever $k$ is negative or $G$ becomes disconnected; note that a disconnected graph cannot be made connected by edge deletions.   We call the algorithm of Proposition~\ref{prop:find-small-fis}, and return ``yes'' if $G$ is already a proper Helly circular-arc graph.  Otherwise, an induced subgraph $F$ is found.  The algorithm calls itself $|E(F)|$ times, each with a new instance $(G - u v, k - 1)$ for some edge $uv$ in $G[F]$.  By Proposition~\ref{prop:find-small-fis}, each recursive call can be made in $O(m + n)$ time.  Therefore, the total running time is $O(10^k\cdot (m + n))$, where $10$ is the number of edges in a $W_{5}$.
\end{proof}

It is straightforward to adapt an approximation algorithm for the proper Helly circular-arc vertex deletion problem from the parameterized algorithm in Theorem~\ref{thm:algs-phcag-deletion}.
\begin{theorem}\label{thm:alg-approximation}
  There is an $O(n m + n^2)$-time approximation algorithm of approximation ratio $6$ for the minimization version of the proper Helly circular-arc vertex deletion problem.
\end{theorem}
\begin{proof}
Let $G$ be the input graph.  We start by creating an empty set $X$.  We call the algorithm in Proposition~\ref{prop:find-small-fis} with $G - X$.  If the algorithm returns a vertex set, then add all of them to $X$.  We continue this process until each component of $G - X$ is a proper Helly circular-arc graph.
We call the $6$-approximation algorithm of Cao~\cite{cao-17-unit-interval-editing} to find a set $Y$ such that $(G - X) - Y$ is a proper interval graph.
We then take $V_-$ to be the smaller of $X\cup Y$ and $V(G)\setminus V(C)$, where $C$ is a component of the maximum order of $G- X$, with ties broken arbitrarily.  Note that we always have $X\subseteq V_-$.

To analyze the approximation ratio, let $V^*_-$ be an optimal solution of $G$.
Since each vertex set returned by Proposition~\ref{prop:find-small-fis} induces a subgraph in $\{ \overline{C_3^*}, C_4^*, S_3, \overline{S_3}, W_4, W_5, \overline{C_6}\}$, the set $V^*_-$ contains at least one vertex from it according to Corollary~\ref{thm:phcag}.  Thus,
\[
  |X| \le 6  |V^*_-\cap X|.
\]
Since $G - X$ is an induced subgraph of $G - (V^*_-\cap X)$, we have
\[
  \mathrm{opt}(G - X)\le \mathrm{opt}(G - (V^*_-\cap X)) \le |V^*_-\setminus X|.
\]
Further, let $Y^*$ be an optimal solution of $G - X$.
If $(G - X) - Y^*$ is a proper interval graph, then
\[
  |V_-\setminus X| \le |Y| \le 6 |Y^*| = 6\, \mathrm{opt}(G - X);
\]
otherwise,
\[
  |V_-\setminus X| = n - |V(C)| - |X| = \mathrm{opt}(G - X).
\]
In summary,
  \[
    |V_-|  = |X| + |V_-\setminus X| \le 6|V^*_-\cap X| + 6\, \mathrm{opt}(G - X)
    \le 6|V^*_-\cap X| + 6 |V^*_-\setminus X|
    = 6 |V^*_-|,
  \]
  and the approximation ratio is $6$.

  We now analyze the running time.  Since each of $\overline{C_3^*}, C_4^*, S_3, \overline{S_3}, W_4, W_5$, and $\overline{C_6}$ has at least $4$ vertices, at most $n/4$ such subgraphs can be detected and deleted in the first phase.  By Proposition~\ref{prop:find-small-fis},  each of them takes $O(m + n)$ time, hence $O(n m + n^2)$ in total.  The approximation algorithm of Cao~\cite{cao-17-unit-interval-editing}  takes another $O(n m + n^2)$ time. The total running time is thus $O(nm + n^2)$.
\end{proof}
  
From Theorem~\ref{thm:algs-phcag-deletion} we can easily derive an \FPT{} algorithm for the combined deletion problem toward proper Helly circular-arc graphs, which allows $k_{1}$ vertex deletions and $k_{2}$ edge deletions.

\begin{theorem}\label{thm:proper helly mixed deletion}
  There is an $O(16^{k_1+k_2}\cdot (m + n))$-time parameterized algorithm for the proper Helly circular-arc deletion problem.
\end{theorem}
\begin{proof}
  This algorithm is very similar to the one used in Theorem~\ref{thm:algs-phcag-deletion}.
  We may assume that the graph is $\{ \overline{C_3^*}, C_4^*, S_3, \overline{S_3}, W_4, W_5, \overline{C_6}\}$-free; otherwise, we break small forbidden induced subgraphs (on at most six vertices) by branching on deleting one of its vertices or edges.  If the graph is already a proper Helly circular-arc graph, we are done.  If it can be made a proper interval graph by deleting $k_1$ vertices and at most $k_2$ edges, then we return ``yes'' in time $O(4^{k_1+k_2}\cdot (m + n))$~\cite{cao-17-unit-interval-editing}.  Otherwise, we have to make a proper Helly circular-arc graph that is not a proper interval graph.  We return yes if a largest component has at least $n - k_1$ vertices, and no otherwise.  The running time is from the branching step, where we need to consider $6$ vertices or $10$ edges.
\end{proof}

It remains to fill in the gap between the constants in Theorems~\ref{thm:algs-phcag-deletion} and~\ref{thm:proper helly mixed deletion} and those announced in Theorems~\ref{thm: main-theorem2}, with an idea is similar to \cite[Proposition 4.4 and Theorem 4.5]{cao-17-unit-interval-editing}. 
Of the seven small forbidden induced subgraphs, only $S_{3}$, $W_{5}$, and $\overline{C_{6}}$ has more than eight edges.
For an $S_{3}$, either we delete one edge between two degree-four vertices, or we have to delete both edges incident to a degree-2 vertex.
For a $W_{5}$ we only need to consider the five edges incident to the center vertex: deleting edge(s) from the $C_{5}$ leaves a claw.
For a $\overline{C_{6}}$, either we delete one edge that is not in any triangle, or we have to delete at least four edges from the two triangles.
Thus, the edge deletion problem and the general deletion problem can be solved in time $O(8^k\cdot (m + n))$ and $O(14^{k_1+k_2}\cdot (m + n))$, respectively.

\section{Proper Helly circular-arc completion}

Compared to the deletion problems, the completion problem toward proper Helly circular-arc graphs is significantly more difficult.  For all the deletion problems, we can always assume that the graph is connected, and then by Corollary~\ref{cor:phcag-small-fis}, we are only concerned with small forbidden induced subgraph.
Since adding edges may make a graph connected, we cannot assume connectedness for the completion problem.

Every hole in a proper Helly circular-arc graph is clearly dominating, and we can be more specific on the intersection between a hole and the neighborhood of any vertex.

\begin{proposition}\label{lem:hole-dominating}
  Let $H$ be a hole in a proper Helly circular-arc graph.  Every vertex in this graph has at least two neighbors on $H$.
\end{proposition}
\begin{proof}
  In any arc representation for the graph, the arcs for vertices on $H$ cover the whole circle.  Thus, every vertex $x$ has a neighbor on $H$.  Moreover, if $x$ has a single neighbor on $H$, then there is a claw, contradicting Corollary~\ref{thm:phcag}.
\end{proof}

It is well known that the maximal cliques of an interval graph can be arranged as a path.
Gavril~\cite{gavril-74-algorithms-cag} showed that the maximal cliques of a Helly circular-arc graph can be arranged as a circle.  This implies that a Helly circular-arc graph has a linear number of maximal cliques.
\begin{theorem}[\cite{gavril-74-algorithms-cag}]
  \label{lem:helly-cliques}
  A graph $G$ is a Helly circular-arc graph if and only if its maximal cliques can be arranged as a circle so that for every vertex $v$ in $G$, the maximal cliques containing $v$ are consecutive.
\end{theorem}

We use a \emph{clique circle} to denote the circular arrangement of maximal cliques specified in Theorem~\ref{lem:helly-cliques}, and a \emph{clique path} is defined analogously.  Note that a clique path can always be viewed as a clique circle, while if two consecutive cliques of a clique circle are disjoint, then it can be viewed as a clique path.   We put the maximal cliques of a clique circle in clockwise order, and for each vertex $v$ that is not universal, let $\lp{v}$ and $\rp{v}$ denote indices of the most anticlockwise, and respectively, the most clockwise clique containing $v$.  We set the arc for the vertex $v$ to be $[\lp{v}, \rp{v}]$.

\begin{proposition}\label{lem:clique-circle-arcs}
  Let $\{K_1, \ldots, K_{\ell}\}$ be a clique circle for a graph $G$.  The set of arcs $\{[\lp{v}, \rp{v}]\mid v\in V(G)\}$ is an arc representation for $G$.
\end{proposition}
\begin{proof}
  Since all the endpoints of the arcs are integral, two arcs intersect, if and only if there is an integer in their intersection.  Note that an integer $i$ is contained in the arcs for vertices $u$ and $v$ if and only if both $u$ and $v$ are contained in $K_i$, hence $u v \in E(G)$.
\end{proof}

In passing we remark that a connected proper Helly circular-arc graph has a unique clique cycle, up to full reversal~\cite{deng-96-proper-interval-and-cag, cao-22-recognition-cag}.

Proper interval graphs are precisely claw-free interval graphs, which can be restated as a graph is a proper interval graph if and only if it is claw-free and has a clique path.  One may thus expect that a graph is a proper Helly circular-arc graph if and only if it is claw-free and has a clique circle.  As we have mentioned, however, $S_3$ is a Helly circular-arc graph and hence has a clique circle, but it is not a proper Helly circular-arc graph even though it is claw-free.
The following statement can be directly observed from forbidden induced subgraphs of the class of proper Helly circular-arc graphs and of the class of normal Helly circular-arc graphs; see also Lin et al.~\cite[Theorem 9]{lin-13-nhcag-and-subclasses}.\footnote{An arc representation is known to be normal and Helly if no set of three or fewer arcs covers the circle \cite{gavril-74-algorithms-cag, cao-17-nhcag}.}

\begin{lemma}[\cite{lin-13-nhcag-and-subclasses}]
  \label{lem:phcag-clique-circle}
  A graph is a proper Helly circular-arc graph if and only if it is claw-free and it has an arc representation in which no set of three or fewer arcs covers the circle.
\end{lemma}

If a proper Helly circular-arc graph $G$ is not an interval graph, then it has a hole.
The structure of every local part of $G$ is very like a proper interval graph.
With the removal of two maximal cliques with no edge in between from $G$, the hole is separated into two sub-paths.  Since every remaining vertex is adjacent to one of the two sub-paths, the remaining graph has precisely two components.

\begin{lemma}\label{lem:isolation}
  Let $G$ be a proper Helly circular-arc graph that is not an interval graph.
  Let $A_1$ and $A_2$ be two maximal cliques of $G$ with no edge between them, and let $B_1$ and $B_2$ be the vertex sets of the two components of $G - (A_1\cup A_2)$.
  Let $G_1$ be any proper interval graph on $B_1\cup A_1\cup A_2$ in which $A_1$ and $A_2$ are the end cliques, and $N_{G_1}(A_i)\cap B_1 = N_{G}(A_i)\cap B_1$ for $i = 1, 2$.
  Replacing $G[B_1\cup A_1\cup A_2]$ by $G_1$ gives another proper Helly circular-arc graph.
\end{lemma}
\begin{proof}
  Let $G'$ denote the resulting graph (after replacing $G[B_1\cup A_1\cup A_2]$ by $G_1$ in $G$).
  We may number the maximal cliques of $G$ in a way that $K_1 = A_1$ and $K_p = A_2$ while $B_1\subseteq \bigcup_{i=1}^p K_i$.  Let $K_1, K'_{2}, \ldots, K'_{q}, K_p$ be a clique path for $G_1$.

  We first argue that $G'$ is claw-free.
  Suppose for contradiction that $G'[\{v, x, y, z\}]$ is a claw, with $v$ being the center.  Since both $G$ and $G_1$ are claw-free, the claw must intersect both $B_1$ and $B_2$.
  Thus, $v$ must be in $K_1$ or $K_p$.  Without loss of generality, assume $v\in K_p$, while $x\in B_1$ and $z\in B_2$.
  Note that the vertex $z$ has the same neighborhood in $G$ and in $G'$.
  Since $G[\{v, x, y, z\}]$ is not a claw, we have $x y\in E(G)\setminus E(G')$.  By the assumption $N_{G_1}(K_p)\cap B_1 = N_{G}(K_p)\cap B_1$, we must have $y\in B_1$.
  But then $G_1[\{v, x, y, z'\}]$, where $z'$ is any vertex in $K_p\setminus K'_{q}$, is a claw, a contradiction to that $G_1$ is a proper interval graph.

  We then argue that
  \[
     K_1, K'_{2}, \ldots, K'_{q}, K_p, K_{p+1}, \ldots
  \]
  gives a clique circle for $G'$.
  Since no vertex in $B_1$ or $B_2$ is adjacent to all vertices in $K_1$ or in $K_p$ in $G'$, both $K_1$ and $K_p$ are maximal cliques of $G'$.
  Since there is no edge between vertices in $B_1$ and $B_2$, both $K'_{i}, 2\le i \le q$ and $K_j, j >p$ are maximal cliques of $G'$.
  A vertex $v\not\in K_1\cup K_p$ is in either $B_1$ or $B_2$, and it is contained in either $K'_{i}, 2\le i \le q$ or $K_j, j >p$.  Since $K_1, K_2, \ldots, K_p, \ldots$ is a clique circle of $G$ and $K_1, K'_{2}, \ldots, K'_{q}, K_p$ is a clique path for $G_1$, those containing $v$ are consecutive.  Now suppose without loss of generality $v\in K_p$.  Maximal cliques containing $v$ are $\ldots, K'_{q}, K_p, K_{p+1}, \ldots$.  Therefore, for any vertex $v$, the maximal cliques of $G'$ containing $v$ are consecutive.

  We can use Proposition~\ref{lem:clique-circle-arcs} to derive an arc representation for $G'$.  For any set of arcs to cover the whole circle in the representation, we need one arc from $B_1$, $B_2$, $A_1$, and $A_2$, hence at least four arcs.  
  Thus, $G'$ is a proper Helly circular-arc graph by Lemma~\ref{lem:phcag-clique-circle}.
\end{proof}

For the completion problem, we may again assume that the input graph $G$ is free of $\overline{C_3^*}$, $C_4^*$, $S_3, \overline{S_3}, W_4, W_5$, and $\overline{C_6}$.
We are done if $G$ is already a \phcag{}.  In particular, this is the case when $G$ is a proper interval graph or when $G$ is connected (Corollary~\ref{cor:phcag-small-fis}).
Thus, we may assume that $G$ is not connected and it is not a proper interval graph.
There must be a hole in $G$ (Corollary~\ref{cor:phcag-uig}), and we add either a chord of this hole, or an edge between this hole and every vertex in other components.  
If there is a hole of length of no more than $16 k + 16$, then there are only $O(k^2)$ such choices, and we can branch on adding one of them.  In the rest every hole is longer than $16 k + 16$.  
Let $H$ be such a hole, and let $G_0$ be the component of $G$ that contains $H$; note that $G_0$ is a proper Helly circular-arc graph.
After adding $k$ or fewer edges, there must be a hole of length greater than $k$ in the subgraph induced by $V(H)$.  Thus, for every vertex $x$ in $V(G)\setminus V(G_0)$, an edge must be added between $x$ and $H$.
We can return ``no'' if $|V(G_0)| < n - {k}$.   Other components have fewer than $k$ vertices while any hole is longer than $16 k + 16$, and thus they are already proper interval subgraphs.

We say that a vertex $x$ is \emph{touched} by a solution $E_+$ if $x$ is an endpoint of an edge in $E_+$, and a set $X$ of vertices is \emph{touched} if at least one vertex in $X$ is touched.  All vertices in $V(G)\setminus V(G_0)$ are touched, and we are more concerned with touched vertices in $G_0$.

\begin{proposition}\label{lem:untouched-maxcliques}
  Let $E_+$ be a solution to $G$.  If a maximal clique $K$ of $G$ is untouched by $E_+$, then $K$ is a maximal clique of $G + E_+$.
\end{proposition}
\begin{proof}
  It is clear that $K$ remains a clique in $G + E_+$.  Suppose that $K$ is not maximal, then it is a proper subset of a maximal clique $K'$ of $G + E_+$.  We can find a vertex $x\in K'\setminus K$.  But then there must be $y\in K$ such that $x y\in E_+$, contradicting that $y$ is untouched.  
\end{proof}

We may fix a clique circle $\mathcal{K}$ of $G_0$, and assume that $H$ and $\mathcal{K}$ are numbered such that no neighbor of $v_1$ or $v_{|H|}$ is touched, and $\{v_{1}, v_{|H|}\}\subseteq K_1$, which is untouched.
Note that $G_0$ has at least $|H|$ maximal cliques, and few of them are touched by a $k$-solution $E_+$.
By Lemma~\ref{lem:isolation} and Proposition~\ref{lem:untouched-maxcliques}, these untouched maximal cliques serve ``isolators'' of the modifications.

We can guess another untouched maximal clique $K_p$ of $G_0$ that is disjoint from and nonadjacent to $K_1$.  Then $H$ is broken into two paths in $G -(K_1\cup K_p)$.  Recall that every vertex in $V(G)\setminus V(G_0)$ needs to be connected to a vertex on $H$.
When the graph $G + E_+$ is not a proper interval graph, every small component is connected to one and only one of the two sub-paths of $H$ (otherwise the resulting graph remains connected after all vertices in $K_1\cup K_p$ removed, contradicting Proposition~\ref{lem:untouched-maxcliques}).
  We can guess in $2^k$ time to which side each small component is attached.  Then we need to add edges to make two proper interval graphs.  
However, we cannot call the proper interval completion problem to solve this task.  For example, the subgraph induced by $\bigcup_{i=1}^p K_i$ together with the small components is already a proper interval graph.  The trouble is how to make them connected while keeping $K_1$ and $K_p$ the ends of the final clique path.  The same holds for the other part of the problem.

For each $i$ with $1< i < |H|$, let the maximal cliques containing $v_i$ be $K_{\mathrm{from}(i)}$, $\ldots$, $K_{\mathrm{to}(i)}$; i.e., $\mathrm{from}(i)$ and $\mathrm{to}(i)$ are the smallest and, respectively, the largest indices.  We define $K_{\mathrm{to}(1)}$ and $K_{\mathrm{from}(|H|)}$ analogously.
Let $r$ denote the number of small components, denoted as $C_1, \ldots, C_r$.
For each pair $p, q$ of indices with $1 < p < q < |H|$,
and each subset $S$ of $[1..r]$, we check whether it is possible to add at most $k$ edges to make $G[\bigcup_{i = \mathtt{to}(p)}^{\mathtt{from}(q)} K_i\cup \bigcup_{j\in S} C_j]$ a proper interval graph, under the condition that $K_{\mathtt{to}(p)}$ and $K_{\mathtt{from}(q)}$ are the end cliques and remain untouched.
Let $\beta(S, p, q)$ denote the minimum cost if it is at most $k$, or $\infty$.
We define $\beta(S, p, q)$ to be $\infty$ when $K_{\mathtt{to}(p)}$ and $K_{\mathtt{from}(q)}$ are not disjoint.
Then $\beta([1..r], 1, |H|)$ is the value we need, which can be calculated as follows.

\begin{proposition}\label{lem:dp}
  The value of $\beta([1..r], 1, |H|)$ can be computed in $k^{O(k)} (n + m)$ time.
\end{proposition}
\begin{proof}
  First, for $a$ and $b$ with $a < b \le b + 8k$ and $S\subseteq [1..r]$, we calculate $\beta(S, a, b)$ as follows.  
  From each component $C_j$ with $j \in S$, we take a vertex $x$, guess a vertex $v_i$ with $a < i <b$, and add the edge $x v_i$.  After that, the subgraph induced by $\bigcup_{i = \mathtt{to}(a)}^{\mathtt{from}(b)} K_i\cup \bigcup_{j\in S}C_j$ is connected.  It will remain connected after adding edges.   We then branch on adding edges to destroy induced subgraphs in $\{ \overline{C_3^*}, S_3, \overline{S_3}\}$ and holes in the subgraph induced by $\bigcup_{i = \mathtt{to}(a)}^{\mathtt{from}(b)} K_i\cup \bigcup_{j\in S}C_j$, without adding any edges incident to $K_{\mathtt{to}(a)}$ or $K_{\mathtt{from}(b)}$. 

  Since $G_0$ is a proper Helly circular-arc graph, a vertex is adjacent to at most four vertices on $H$ (there is a claw otherwise).  Therefore, at most $8 k$ vertices on $H$ have touched neighbors.
  If $v_{b-1}$ has a touched neighbor, then for some $i$ with $2\le i \le 8 k + 2$,
  the vertex $v_{b - i}$ has no touched neighbor.
For $b - a > 8k$, by Lemma~\ref{lem:isolation}, we have 
\begin{equation}
  \label{eq:dp}
  \beta(S, a, b) =
    \min_{\substack{1\le i \le 8 k + 1 \\ S'\subseteq S}} \left( \beta(S\setminus S', a,  b - i) + \beta(S',  b - i, b) \right).
\end{equation}
We can then use dynamic programming to calculate  $\beta([1..r], 1, |H|)$ with \eqref{eq:dp}.
\end{proof}

We are now ready to summarize the algorithm in Figure~\ref{fig:alg-phcag}.

\begin{figure}[h!]
  \centering 
  \begin{tikzpicture}
    \path (0,0) node[text width=.9\textwidth, inner sep=10pt] (a) {
      \begin{minipage}[t!]{\textwidth}

        \begin{tabbing}
          Aa\=aA\=Aa\=MMMMAAAAAAAAAAAAa\=A \kill
          1.\> \textbf{if} there exists an induced subgraph $X$ in $\{ \overline{C_3^*}, C_4^*, S_3, \overline{S_3}, W_4, W_5, \overline{C_6}\}$ \textbf{then}
          \\
          \>\> \textbf{branch} on adding missing edges of $X$; \codecomment{returns ``no'' if $k$ becomes negative.}
          \\
          2.\> {\bf if} $G$ is a proper Helly circular-arc graph {\bf then return} ``yes'';
          \\
          3.\> find a hole $H$ of $G$ and let $G_0$ be the component of $G$ that contains $H$;
          \\
          4. \> \textbf{if} $|H| \le 16 k + 16$ \textbf{then}
          \\
          \>\> \textbf{branch} on adding chords of $H$ or edges between $H$ and other components;
          \\
          5.\> {\bf if} $|V(G_0)| < n - {k\over 2}$ {\bf then return} ``no'';
          \\
          6.\> \textbf{if} $\beta([1..r], 1, |H|) \le k$ \textbf{then return} ``yes'';
          \\
          \> \textbf{else return} ``no.''
        \end{tabbing}  
      \end{minipage}
    };
    \draw[draw=gray!60] (a.north west) -- (a.north east) (a.south west) -- (a.south east);
  \end{tikzpicture}
  \caption{The outline of the algorithm for the proper Helly circular-arc completion problem.}
  \label{fig:alg-phcag}
\end{figure}

\begin{lemma}
  The proper Helly circular-arc completion problem can be solved in time $k^{O(k)}  (n + m)$.
\end{lemma}
\begin{proof}
  If $G$ contains an induced subgraph in $\{ \overline{C_3^*}, C_4^*, S_3, \overline{S_3}, W_4, W_5, \overline{C_6}\}$, then any solution must contain at least one missing edge in this subgraph.  After step~1, $G$ is $\{ \overline{C_3^*}, C_4^*, S_3, \overline{S_3}, W_4, W_5, \overline{C_6}\}$-free, and it is already a proper Helly circular-arc graph if it is connected (Corollary~\ref{cor:phcag-small-fis}).  Since every component of $G$ is a proper circular-arc graph, but $G$ is not a proper interval graph, there must be a hole in $G$ (Corollary~\ref{cor:phcag-uig}).  
  Step~4 is obviously correct: a solution contains a chord of $H$ or an edge between $H$ and other components of $G$.  Now that $H$ is longer than $16 k + 16$, there must be a hole after adding at most $k$ edges.  Every vertex in $V(G)\setminus V(G_0)$ must be made adjacent to this hole.
  The correctness of step~5 follows from Proposition~\ref{lem:hole-dominating}, and the correctness of step~6 follows from the definition of $\beta$.  
  
  We now analyze the running time of the algorithm.  The algorithm makes at most $O(6^k)$ recursive calls in step~1, all with parameter value $k - 1$.  We can check whether $G$ is a proper Helly circular-arc graph in linear time~\cite{cao-17-nhcag}.  When it is not, we can find a hole in linear time.
  At most $k^{O(k)}$ recursive calls are made in step~4, all with parameter value $k - 1$.
  In step~6, we use Proposition~\ref{lem:dp} to calculate $\beta([1..r], 1, |H|)$, for which we need to guess a numbering of $H$ and $\mathcal{K}$.
  At most $8 k$ vertices on $H$ have touched neighbors.  Thus, by trying $16 k$ vertices on $H$, we can find one pair of adjacent vertices as $v_1$ and $v_{|H|}$.
  The total time is thus $k^{O(k)}  (n + m)$.
\end{proof}

A natural question is on the complexity of the proper Helly circular-arc completion problem on $\{ \overline{C_3^*}, C_4^*, S_3, \overline{S_3}, W_4, W_5, \overline{C_6}\}$-free graphs.  A related and (maybe) simpler question is when all but one of the components of the input graph are proper interval graphs, while the other is a proper Helly circular-arc graph.  Let us remark that a polynomial-time algorithm for this later problem implies a polynomial-time algorithm for the proper interval completion problem on proper Helly circular-arc graphs.

However, it is not obvious how to solve the proper Helly circular-arc editing problem, which allows $k_{1}$ vertex deletions, $k_{2}$ edge deletions, and $k_3$ edge additions.
We do not have the optimal substructures as we have seen in the completion problem: we may need to cut some vertices off $G_0$ by deleting edges, and attach them to other parts of $G_0$; see Figure~\ref{fig:bad-example} for an example.  Let us mention that the idea does solve the modification problem to proper Helly circular-arc graphs that allows both vertex deletions and edge additions.

\begin{figure}[h]
  \centering
  \begin{tikzpicture}[every node/.style={empty vertex}, every path/.style={thick}]
    \foreach \i in {1, ..., 12} 
    \draw ({30*\i-15}:2) -- ({30*\i+15}:2);
    \foreach[count =\i] \c in {3, 3, 3, 3, 99, 2, 4, 99, 3, 3, 3, 3} 
    \node[empty vertex] (u\i) at ({-75 - 30*\i}:2) {$\c$};

    \begin{scope}
      \foreach \i in {7, ..., 12} 
      \draw ({30*\i-30}:3) -- ({30*\i}:3);
      \foreach \i in {6, ..., 12} 
      \node[empty vertex] (x\i) at ({30*\i}:3) {$10$};
    \end{scope}

    \node[fill=gray!70] (x) at (0, -1.25) {$1$};
    \draw[] (u1) -- (x) -- (u12);
  \end{tikzpicture}
  \caption{The input graph $G$ comprises two components.  Each node denotes a set of true twins with the designated number of vertices, and a line between two nodes indicates all the edges between the two sets of vertices.  We can make $G$ a proper Helly circular-arc graph by deleting $6$ edges and adding $426$ edges (i.e., $k_1 = 0$, $k_2 = 6$, and $k_3 = 426$).   We need to remove all the six edges incident to the shadowed vertex and add edges between it and the six vertices (in two sets of sizes $2$ and $4$ respectively) at the top.}
  \label{fig:bad-example}
\end{figure}

\begin{conjecture}
  The proper Helly circular-arc editing problem is fixed-parameter tractable.
\end{conjecture}

\section{Proper circular-arc vertex deletion}

Since we will use properties of both the graph $G$ and its complement,
we beg the reader's attentiveness in reading this section.
There are algorithms for the vertex deletion problem toward both proper interval graphs and bipartite permutation graphs.  We are henceforth focused on graphs that are both connected and co-connected.  As usual, we can get rid of small forbidden induced subgraphs easily.

\begin{definition}
We say that a graph is \emph{reduced} if it is both connected and co-connected, and it contains no $C_3^*$, $C_5^*$, $\overline{C_4^*}$, $C_6$, $S_3$, $\overline{S_3^*}$, $F_{1}$, $F_{2}$, $F_{3}$, or $F_{4}$.  
\end{definition}

Similar to Proposition~\ref{prop:find-small-fis}, one can make an algorithm for finding one of the subgraphs list above when the input graph is not reduced.  
We omit details since it does not improve our main algorithm.
The next lemma is complement and similar to Lemma~\ref{lem:small-subgraphs-free}.

\begin{lemma}\label{lem:cycles-in-reduced}
 A reduced graph is $\{C_\ell^* \mid \ell \ge 7\}$-free.
\end{lemma}
\begin{proof}
  Let $R$ be a reduced graph.
  Suppose for contradiction that there exist a hole $H$ of length at least seven and a vertex that is nonadjacent to any vertex on $H$.  Since ${R}$ is connected, we can find a vertex $x$ adjacent to $H$, and another vertex $y$ that is adjacent to $x$ but not to $R$.
  Let $H = v_1 v_2 \cdots v_{|H|}$. 
  We argue first that $x$ cannot be adjacent to two consecutive vertices on $H$.
  Suppose for contradiction that $x$ is adjacent both $v_1$ and $v_2$.
  Then $v_1,v_2$, and $x$ form a triangle.  Since $R$ is free of $C^*_3$, it follows that $x$ is adjacent to both $v_4$ and $v_5$.  But then $v_1,v_2, v_4, v_5$, and $x$ induce a $\overline{C_4^*}$.

  Assume without loss of generality that $x$ is adjacent to $v_3$.  Note that $x$ is adjacent to neither $v_2$ nor $v_4$.
  If $x$ is adjacent to $v_5$ as well, then $x$ is nonadjacent to $v_6$, and $\{v_2, \ldots, v_6, x, y\}$ induces an $F_2$.  By symmetry, $x$ cannot be adjacent to $v_1$ either.  But then  $\{v_1, \ldots, v_5, x, y\}$ induces an $F_1$.
\end{proof}
By definition, a reduced graph is $C_3^*$- and $C_5^*$-free.  

\begin{corollary}\label{lem:reduced-no-odd-cycle-star}
  A reduced graph is $\{C_{2 \ell + 1}^* \mid \ell \ge 1\}$-free.
\end{corollary}

By Theorem~\ref{thm:proper-forbidden-induced-subgraphs}, the definition of reduced graphs, and Lemma~\ref{lem:small-subgraphs-free},
a reduced graph is the complement of a proper circular-arc graph if and only if it does not contain any even hole of length at least eight.
We will therefore be focused on long even holes.
The following statement characterizes reduced graphs that contain long even holes.

\begin{lemma}\label{lem:reduced}
  If a reduced graph contains an even hole of length at least eight, then it is bipartite.
\end{lemma}
\begin{proof}
  Let $v_1 v_2 \cdots v_{\ell}$ be an even hole with $\ell \ge 8$ of a reduced graph $R$, and denote it by $B$.  We prove the lemma with a sequence of claims.

  Our first claim is that no vertex on $B$ participates in any triangle.
  If there is a triangle involving two vertices from $B$, say, $u v_{1} v_{2}$, then either it forms a $C_3^*$ with $v_4$ or $V_5$, when $u$ is not adjacent to one of $v_{4}$ and $v_{5}$, or $\{u, v_{1}, v_{2}, v_{4}, v_{5}\}$ induces a $\overline{C_{4}^*}$.
  Since neither is possible,  we can conclude that there is no triangle involving two vertices of $B$.  Now suppose that there is a triangle involving a single vertex from $B$, say, $u w v_{1}$ with $u, w\in V(R)\setminus V(B)$.   Since $R$ is $C_3^*$-free, for $i = 3, 4$, there is at least one edge between $v_{i}$ and $u, w$.  By the argument above, neither $u$ nor $w$ can be adjacent to both $v_{3}$ and $v_{4}$.  Thus,  $\{u, w, v_{3}, v_{4}\}$ induces a $C_{4}$, and so does the set $\{u, w, v_{6}, v_{7}\}$.  But then $\{u, w, v_{1}, v_{3}, v_{4}, v_{6}, v_{7}\}$ induces an $F_{4}$, a contradiction to that $R$ is reduced. 
  
  The second claim is that if some odd hole of $R$ intersects $B$, then there exists an odd hole of $R$ whose intersection with $B$ is a nonempty sub-path of $B$.
  Let $C$ have the minimum number of vertices out of $B$ among all the odd holes that intersect $B$ (i.e., $|V(C)\setminus V(B)|$ is minimized).
  We take a nontrivial sub-path $P_{C}$ of $C$ such that only the two ends of this path are on $B$.
  Without loss of generality, assume the ends of $P_{C}$ are $v_{1}$ and $v_{p}$.  Let $P_{B}$ be the sub-path $v_{1}v_{2} \cdots v_{p}$ of $B$.  If $P_{B}$ and $P_{C}$ have different parities (one even and the other odd), then they together form an odd induced cycle.  By the first claim, this induced cycle must be an odd hole, and we are done.
  Otherwise, we replace $P_{C}$ by $P_{B}$ in $C$ to produce another odd closed walk.  Therefore, we always end with an odd closed walk that avoids at least one vertex in $V(C)\setminus V(B)$ and whose vertices are all from $V(B)\cup V(C)$.
  From this odd closed walk we can find an odd induced cycle $C'$.
  Since $V(C')\subseteq V(C)\cup V(B)$ and $C'\ne C$, the cycle $C'$ must intersect $B$.  On the other hand, since $B$ is an even hole, $C'$ intersects $C$ as well.
  Since $P_C$ is a nontrivial path, it has at least one inner vertex, which is not on $C'$.  It follows that $|V(C')\setminus V(B)| < |V(C)\setminus V(B)|$, contradicting the selection of $C$.

  The third claim is that if $V(C)\cap V(B)$ is consecutive for an odd induced cycle $C$, then $|V(C)\cap V(B)| \le 4$.  Assume without loss of generality that the intersection of $B$ and $C$ is the sub-path $v_1 v_2 \cdots v_{c}$.  We consider the odd cycle formed by the other \stpath{v_1}{v_c} on $B$ and the other \stpath{v_1}{v_c} on $C$, which consists of  vertices $V(B)\cup V(C) \setminus \{v_2, v_3, \ldots, v_{c-1}\}$.
  From this odd cycle we can retrieve an odd induced cycle $C'$.  If $c > 4$, then $v_{3}$ is not adjacent to any vertex on $C'$.  But then $C'$ and $v_{3}$ form a $C_{|C'|}^*$, contradicting Corollary~\ref{lem:reduced-no-odd-cycle-star}.

  The fourth claim is that no odd induced cycle can intersect $B$.
  Suppose for contradiction that there are such odd induced cycles.  
 We take an odd induced cycle $C$ such that (1) $V(B)\cap V(C)$ induces a path; (2) $|V(B)\cap V(C)|$ is maximized among all cycles satisfying (1); and (3) $C$ is the shortest among all cycles satisfying both (1) and (2).  The existence of such a cycle is ensured by the second claim.
 Let $q = |C|$ and $c = |V(B)\cap V(C)|$.
 Assume without loss of generality that the intersection of $B$ and $C$ is the sub-path $v_1 v_2 \cdots v_{c}$, and we may number the vertices on $C$ as $u_1, u_2, \ldots, u_{q}$ such that $u_i = v_i$ for $i = 1, \ldots, c$.  We consider the cycle $A$ formed by the two \stpath{v_1}{v_c}s $v_c v_{c +1} \cdots v_{\ell} v_1$ and $u_c u_{c +1} \cdots u_{q} u_1$.
The length of $A$, which visits $V(B)\cup V(C) \setminus \{v_2, v_3, \ldots, v_{c-1}\}$,  is $\ell + q - 2 (c - 1)$,  hence odd.
We use chords $v_i u_j$ of $A$ to decompose $A$ into a sequence of induced cycles as follows.
We start with $i = c+2$ and $j$ being the smallest index such that $v_{c+2}u_j\in E(G)$.  The first cycle is $v_{c+2}v_{c+1} u_{c} u_{c+1}\cdots u_{j}$.
See Figure~\ref{fig:lem-reduced}(a) for an illustration.

\begin{enumerate}[{Case} 1,]
\item the first cycle $v_{c+2}v_{c+1} u_{c} u_{c+1}\cdots u_j$ is odd; let it be $C'$.  Note that ${j}$ has the same parity as $c$.  

  First, suppose that $C'$ is induced.  The cycle $C'$ has three common vertices with $B$.  Thus, $c \ge 3$ by the second condition in the selection of $C$.  Note that $v_1$ is not on $C'$.  Since $v_2$ has a neighbor on $C'$ (Corollary~\ref{lem:reduced-no-odd-cycle-star}), the neighbor has to be $v_3$.  Thus, $c = 3$.
  Since $\ell \ge 8$, the only possible neighbor of $v_{1}$ on $C'$ is $u_q$.  Thus, $j = q$.
  If $q = 5$, then $v_{6}$ must be adjacent to $u_{4}$, (the vertex $v_{6}$ is adjacent to some vertex on $C$, while $u_{5}$ cannot form a triangle with $v_{5}$ and $v_{6}$,) and the odd hole $u_{4} v_{3} v_{4} v_{5} v_{6}$ forms a $C_{5}^*$ with $v_{1}$.  Now that $q > 5$, the set $\{v_{1}, v_{2}, v_{3}, v_{4}, v_{5}, u_{4}, u_{5}\}$ induces a long claw.

  In the rest, $C'$ is not induced.  By the definition of $j$, there is no chord of $C'$ incident to $v_{c+2}$.  In other words, all the chords of $C'$ are incident to $v_{c+1}$.  They partition $C'$ into induced cycles, which are all strictly shorter than $C$.  Since $C'$ is odd, at least one of these induced cycles is odd; let it be $C''$.  The intersection of $C''$ and $B$ is one of $\{v_{c+1}\}$, $\{v_{c+1}, v_{c}\}$, and $\{v_{c+1}, v_{c+2}\}$, hence consecutive.
  By the selection of $C$, we must have $|V(B)\cap V(C'')| < |V(B)\cap V(C)| = c$.  But then $v_{c - 2}$ or $v_{c - 1}$ (when $c \ge 3$) or $v_{c - 1}$ (when $c = 2$) has no neighbor on $C''$, contradicting Corollary~\ref{lem:reduced-no-odd-cycle-star}.    
\end{enumerate}

Now that the first cycle is even, we continue as follows.  We find the next neighbor of $v_i$ in $u_{j+1}, \ldots, u_{q}$, if one exists.  Otherwise, we find the edge between $v_{i+1}, \ldots, v_{\ell}$ and $u_{j}, \ldots, u_{q}$.  We check the induced cycle formed by the two chords and the two sub-paths on $B$ and $C$.  We stop if it is odd.  Otherwise we update $i$ and $j$, and continue the search until we reach $v_\ell$.

\begin{enumerate}[{Case} 1,]
  \setcounter{enumi}{1}
\item the cycle $v_{i} \cdots v_{i'} u_{j'} u_{j' - 1} \cdots u_{j}$ is odd, where $v_i u_j$ and $v_{i'}u_{j'}$ are the previous chord and new chord, respectively. Let the cycle be denoted as $C'$.  Note that neither $v_{1}$ nor $v_{c}$ is on $C'$.
  If $c\ge 3$, then $v_{2}$ has no neighbor on $C'$ (because both $B$ and $C$ are induced), contradicting Corollary~\ref{lem:reduced-no-odd-cycle-star}.  
  By the second condition in the selection of $C$, we have $|V(C')\cap V(B)| \le |V(C)\cap V(B)| = c \le 2$.
  Thus, $i' \le i + 1$.
  By Corollary~\ref{lem:reduced-no-odd-cycle-star}, the vertex $v_c$ has at least one neighbor on $C'$.
  \begin{itemize}
  \item Subcase 2.1, $c = 2$.
    See Figure~\ref{fig:lem-reduced}(b).   Note that $v_2$ has three neighbors on $B$ and $C$, of which neither $v_1$ nor $v_3$ can be on $C'$.  (In particular, $v_3$ is on the first cycle, which contains $v_c$, i.e., $v_2$.)  Thus, $j = c + 1 = 3$, which further implies $i = 4$.  Note that $u_q$ cannot be adjacent to $v_4$: if $u_q v_4 \in E(R)$, then $u_q v_3 \not\in E(R)$ because of the first claim, and then $u_q v_1 v_2 v_3 v_4$ is an induced $C_5$, contradicting the second condition in the selection of $C$.
    Again, $v_1$ has a neighbor on $C'$ by Corollary~\ref{lem:reduced-no-odd-cycle-star}, and it has to be $u_q$.  It follows that $i' = 5$, and then by the selection of the cycle, there is no edge between $v_4, v_5$ and any vertex in $\{u_4, u_5, \ldots, u_{q-1}\}$.
    We argue that $v_3$ is not adjacent to $u_4$ by contradiction.  If $q = 5$, then the cycle $v_1 v_2 v_3 u_4 u_5$ intersects $B$ with three vertices; it cannot be induced by the selection of $C$, but then there is a triangle intersecting $B$, contradicting the first claim.   If $q \ge 7$, then $\{v_{1}, v_{2}, v_{3}, v_{4}, v_{5}, u_{4}, u_{5}\}$ induces an $F_{1}$.  Therefore, $v_3 u_4\not\in E(R)$ and the set $\{v_{1}, v_{2}, v_{3}, v_{4}, v_{5}, u_{3}, u_{4}\}$ induces an $F_{2}$.
  \item Subcase 2.2, $c = 1$. 
    See Figure~\ref{fig:lem-reduced}(c).
    Then $|V(C')\cap V(B)| = |V(C)\cap V(B)| = 1$.
    By the third condition in the selection of $C$, we have $|C'| = |C| = q$.  Thus, $j = c+1 = 2$ and $j' = q$ and $i = c+2 = 3$.
    None of $v_2$, $v_4$, and $v_\ell$ can be adjacent to $u_2$ by the first claim.
    We now argue that $u_3$ cannot be adjacent to $v_2$, $v_4$, or $v_\ell$.  Suppose that $u_3 v_2\in E(R)$.  Then $v_1 v_2 u_3 u_4 \cdots u_q$ is an odd cycle, from which we can find an induced odd cycle.  But this contradicts the selection of $C$: either it shares two vertices with $B$, or it is shorter than $C$.  It is likewise when $u_3 v_4\in E(R)$ or $u_3 v_\ell\in E(R)$: in particular, we consider the odd cycles $v_3 v_4 u_3 u_4 \cdots u_q$, and respectively, $v_1 v_\ell u_3 u_4 \cdots u_q$.   The set $\{v_{\ell}, v_{1}, v_{2}, v_{3}, v_{4}, u_{2}, u_{3}\}$ induces an $F_{2}$.
  \end{itemize}
\end{enumerate}

If neither of the previous two cases is true, then the last cycle must be odd.  All the small cycles use edges from $A$ and the chords of $A$.
  Moreover, each edge on $A$ appears in precisely one of the cycles, and if a chord is used, then it appears in precisely two of the cycles.  Therefore, the sum of the lengths of all these cycles has the same parity as $A$, hence odd.

\begin{enumerate}[{Case} 1,]
  \setcounter{enumi}{2}
\item the last cycle $C'$, which contains $v_{1}$, is odd.  Since $R$ is $C_\ell^*$-free (By Lemma~\ref{lem:cycles-in-reduced} and note that $\ell > 7$), the vertex $u_{q-1}$ is adjacent to some $v_{i}$ with $c <i \le \ell$.
  The vertex $u_q$ may or may not have other neighbors (other than $v_1$) on $B$. 
  Thus, $C'$ is either $v_{1} u_{q} v_{i} v_{i+1} \cdots v_{\ell}$ or $v_{1} u_{q} u_{q - 1} v_{i} v_{i+1} \cdots v_{\ell}$.
  By the selection of $C$, we have $|V(C')\cap V(B)| = \ell - i + 2 \le c$.
  In particular, $c \le 3$, then $i \ge \ell - 1$.
  Therefore, $v_{3}$ is not adjacent to any vertex on $C'$, contradicting Corollary~\ref{lem:reduced-no-odd-cycle-star}.
\end{enumerate}
Thus, no vertex on $B$ is contained in any odd induced cycle.

Finally, we show that $R$ does not contain any odd cycle at all.  Let $C$ be an odd induced cycle that is disjoint from $B$.  First assume $C = x_{1} x_{2} x_{3}$.  Since $R$ is $C_{3}^*$-free, every vertex on $B$ is adjacent to at least one vertex on $C$.  Without loss of generality that $x_{1}$ has the largest number of neighbors on $B$, and let their indices be $i_{1}$, $i_{2}$, $\ldots$, $i_{p}$, sorted increasingly.  Note that all of them have the same parity by the fourth claim.
Since $R$ is $C_{6}$-free, $i_{j+1} - i_{j}$ is either two or at least six for all $j = 1, \ldots, p-1$.  Since $p\ge \ell/3$, there must be three consecutive ones with differences two; assume without loss of generality, that they are $v_{1}$, $v_{3}$, and $v_{5}$.  If $\ell = 8$, then $x_{1} v_{5}v_{6} \cdots v_{8}v_{1}$ has length six, and hence $x_{1}$ must be adjacent to  $v_{7}$ as well; otherwise, $x_{1}$ has another neighbor on $B$ because $p\ge \ell/3 \ge 4$.  This neighbor forms an $F_{3}$ with $\{x_{1}, v_{1}, v_{2}, \ldots, v_{5}\}$.
Now that $|C| \ge 5$; let it be $x_{1} x_{2} \cdots x_{|C|}$.  We take $v_{i}\in N(x_{1})\cap V(B)$ and $v_{j}\in N(x_{3})\cap V(B)$.  
The sub-path $v_{i} v_{i+1} \cdots v_{j}$ forms an odd cycle with either $x_{1} x_{2} x_{3}$ or $x_{3} x_{4}\cdots x_{|C|} x_{1}$.  From this odd cycle we can retrieve an induced odd cycle, which has to intersect both $B$ and $C$.  This contradicts the fourth claim, and concludes the proof of this lemma.
\end{proof}

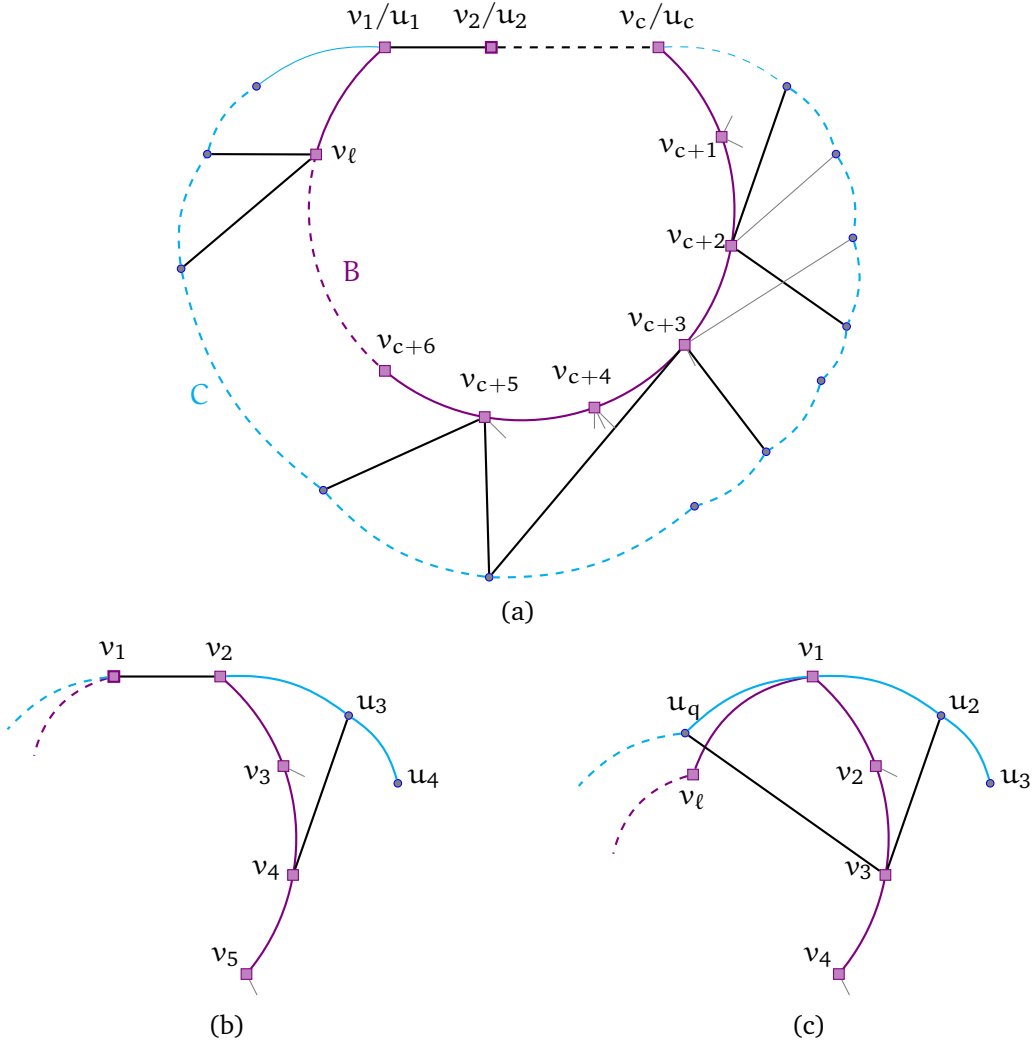
\begin{figure}[h!]
  \centering 
  \begin{subfigure}[b]{0.8\linewidth}
    \centering
    \begin{tikzpicture}[scale=.7]
      \pgfmathsetmacro{\radius}{4.}
      \draw[violet, thick] ({230}:\radius) arc ({230}:{410}:\radius);
      \draw[violet, thick] (130:\radius) arc ({130}:{165}:\radius);
      \draw[violet, thick, dashed] ({165}:\radius) arc ({165}:{230}:\radius);    
      \foreach \i in {1, ..., 6} {
        \node[uvertex] (vc+\i) at ({50 - \i*30}:\radius) {};
        \node at ({50 - \i*30}:{\radius-.65}) {$v_{c+\i}$};
      }

      \node[violet] at (200:{\radius-.6}) {$B$};
      \node[cyan] at (210:{\radius+3}) {$C$};
      \foreach \a/\l in {50/c, 130/1} {
        \node[uvertex, "$v_{\l}/u_{\l}$"] (v\l) at (\a:\radius) {};
      }
      \node[uvertex, "$v_{\ell}$" right] (vl) at (165:\radius) {};
      \draw[thick]  (v1) -- ++(2, 0) node [uvertex, "$v_2/u_2$"] (v2) {};
      \draw[thick, dashed]  (v2) -- (vc);

      \foreach[count =\i] \a/\r/\l in {25/5.5/u_{\beta'(1)}, 10/6/, -5/6.25/, -20/6.5/u_{\beta''(1)}, -30/6.5/, -45/6.5/u_{\beta'(2)}, -60/6.5/, -95/7/u_{\beta''(2)}, -125/6.5/u_{\beta'(3)}, -170/6.5/u_{\beta'(.)}, 170/6/u_{\beta''(.)}, 155/5.5/u_q} {
        \node[filled vertex] (u\i) at (\a:\r) {};
      }
      \foreach[count =\from] \to in {2,3, ..., 12}
      \draw [thick, cyan, dashed, bend left=20] (u\from) edge (u\to);
      \draw [cyan, bend left=20]  (vc) edge[dashed] (u1) (u12) edge[bend left] (v1);

      \draw[thick] (vc+2) --  (u1) (vc+2) --  (u4) (vc+3) --  (u6) (vc+3) --  (u8) (vc+5) --  (u8) (vc+5) --  (u9) (vl) --  (u10) (vl) --  (u11);
      \draw[gray] (vc+1) -- ++(.4, -.2) (vc+1) -- ++(.2, .4) (vc+3) -- ++(.2, -.4) (vc+4) -- ++(.4, -.4) (vc+4) -- ++(0, -.4) (vc+3) --  (u3) (vc+2) --  (u2) (vc+5) -- ++(.4, -.4) (vc+4) -- ++(.2, -.4);
    \end{tikzpicture}
    \caption{}
  \end{subfigure}

  \begin{subfigure}[b]{0.45\linewidth}
    \centering
    \begin{tikzpicture}[scale=.7]
      \pgfmathsetmacro{\radius}{4.}
      \draw[violet, thick] ({320}:\radius) arc ({320}:{410}:\radius);
      \foreach[count = \j from 3] \i in {1, ..., 3} {
        \node[uvertex] (vc+\i) at ({50 - \i*30}:\radius) {};
        \node at ({50 - \i*30}:{\radius-.5}) {$v_{\j}$};
      }
      \foreach \a/\l in {50/2} {
        \node[uvertex, "$v_{\l}$"] (v\l) at (\a:\radius) {};
      }
      \draw[thick]  (vc) -- ++(-2, 0) node[uvertex, "$v_1$"] (v1) {};
      \draw[violet, dashed, thick] (v1) edge[bend right] ++(-1.5, -1.5);

      \foreach[count =\i] \a/\r/\l in {25/5.5/u_3, 10/6/u_4} {
        \node[filled vertex] (u\i) at (\a:\r) {};
        \node at (\a:{\r+.5}) {$\l$};
      }
      \foreach[count =\from] \to in {2}
      \draw [thick, cyan, bend left=20] (u\from) edge (u\to);
      \draw [thick, cyan, bend left=20]  (vc) edge (u1) (v1) edge[bend right, dashed] ++(-2, -1);

      \draw[thick] (vc+2) --  (u1);
      \draw[gray] (vc+1) -- ++(.4, -.2) (vc+3) -- ++(.2, -.4);
    \end{tikzpicture}
    \caption{}
  \end{subfigure}
  \quad
  \begin{subfigure}[b]{0.45\linewidth}
    \centering
    \begin{tikzpicture}[scale=.7]
      \pgfmathsetmacro{\radius}{4.}
      \draw[violet, thick] ({320}:\radius) arc ({320}:{410}:\radius);
      \foreach[count = \j from 2] \i in {1, ..., 3} {
        \node[uvertex] (vc+\i) at ({50 - \i*30}:\radius) {};
        \node at ({50 - \i*30}:{\radius-.5}) {$v_{\j}$};
      }
      \foreach \a/\l in {50/1} {
        \node[uvertex, "$v_{\l}$"] (v\l) at (\a:\radius) {};
      }
      \node[uvertex, "$v_{\ell}$" below] (vl) at (75:1.25) {};
      \draw[violet, thick] (v1) edge[bend right] (vl)  (vl) edge[bend right, dashed] ++(-1.5, -1.5) ;

      \foreach[count =\i] \a/\r/\l in {25/5.5/u_2, 10/6/u_3} {
        \node[filled vertex] (u\i) at (\a:\r) {};
        \node at (\a:{\r+.5}) {$\l$};
      }
      \foreach[count =\from] \to in {2}
      \draw [thick, cyan, bend left=20] (u\from) edge (u\to);
      \node[filled vertex, "$u_{q}$"] (uq) at (85:2) {};    
      \draw [thick, cyan, bend left=20]  (vc) edge (u1) (v1) edge[bend right] (uq) (uq) edge[bend right, dashed] ++(-2, -1);

      \draw[thick] (vc+2) --  (u1) (vc+2) --  (uq);
      \draw[gray] (vc+1) -- ++(.4, -.2) (vc+3) -- ++(.2, -.4);
    \end{tikzpicture}
    \caption{}
  \end{subfigure}

  \caption{Illustration for the proof of Lemma~\ref{lem:reduced}.}
  \label{fig:lem-reduced}
\end{figure}

The following is immediate from Theorem~\ref{thm:proper-forbidden-induced-subgraphs}, Lemma~\ref{lem:small-subgraphs-free}, and Lemma~\ref{lem:reduced}.
\begin{corollary}\label{cor:reduced-non-bipartite}
  If a reduced graph $R$ is not bipartite, then $R$ is the complement of a proper circular-arc graph.  
\end{corollary}

We are now ready to present the algorithm for the proper circular-arc vertex deletion problem.
Let $(G, k)$ be an instance to the problem, and we may assume without loss of generality that $G$ does not contain any  small forbidden induced subgraphs on at most seven vertices.
If there is a set $V_-$ of $k$ vertices such that $G - V_-$ is a proper interval graph or $\overline G - V_-$ is a bipartite permutation graph, then we are done.
Hence, we will look for a solution $V_-$ such that $G - V_-$ is both connected and co-connected (Propositions~\ref{lem:connected} and~\ref{lem:co-bipartitie}).

For this purpose we may assume that $G$ itself is connected and co-connected: if $G$ is not connected, we can work on the components of $G$ one by one, and it is similar for $\overline G$.  Thus, $\overline G$ is a reduced graph.
If $\overline G$ is not bipartite, then $G$ is already a proper circular-arc graph (Corollary~\ref{cor:reduced-non-bipartite}).  Otherwise, $\overline G$ is bipartite, of which any induced subgraph of it is bipartite.  In other words, if there exists a solution $V_-$, then $\overline G - V_-$ is a bipartite permutation graph, and this has been handled already.

\begin{figure}[h!]
  \centering 
  \begin{tikzpicture}
    \path (0,0) node[text width=.8\textwidth, inner sep=10pt] (a) {
      \begin{minipage}[t!]{\textwidth}

        \begin{tabbing}
          AAA\=Aaa\=aaa\=Aaa\=MMMMMAAAAAAAAAAAA\=A \kill
          1.\> \textbf{if} $(G, k)$ is a yes-instance of proper interval vertex deletion \textbf{then}
          \\
          \>\>\textbf{return} ``yes'';
          \\
          2.\> \textbf{if} $(\overline G, k)$ is a yes-instance of bipartite permutation vertex deletion \textbf{then}
          \\
          \>\>\textbf{return} ``yes'';
          \\
          \> \codecomment{We're looking for a solution $V_-$ with both $G - V_-$ and $\overline G - V_-$ connected.}
          \\
          3.\> \textbf{branch} on deleting vertices of small forbidden induced subgraphs;
          \\
          4.\> \textbf{guess} a component $C$ of $G$;
          \\
          5.\> $k \leftarrow k - |V(G)\setminus V(C)|$;
          \\
          6.\> \textbf{if} $\overline C$ is bipartite \textbf{then return} ``no'';
          \\
          7.\> $C'\leftarrow$ a maximum non-bipartite component of $\overline C$;
          \\
          8.\> \textbf{if} $|V(C)\setminus V(C')| \le k$ \textbf{then return} ``yes'';
          \\
          \> \textbf{else return} ``no.''
        \end{tabbing}
      \end{minipage}
    };
    \draw[draw=gray!60] (a.north west) -- (a.north east) (a.south west) -- (a.south east);
  \end{tikzpicture}
  \caption{The outline of the algorithm for proper circular-arc vertex deletion.}
  \label{fig:pcag-deletion}
\end{figure}

Again, it is quite straightforward to turn this algorithm into an approximation algorithm, and the proof is similar to that of Theorem~\ref{thm:alg-approximation}.
\begin{theorem}\label{main-theorem}
  There is a $9^{k}\cdot n^{O(1)}$-time parameterized algorithm for the proper circular-arc vertex deletion problem, and 
  a polynomial-time approximation algorithm of approximation ratio~$9$ for the minimization version of the proper circular-arc vertex deletion problem.
\end{theorem}
\begin{proof}
  We use the parameterized algorithm described in Figure~\ref{fig:pcag-deletion}.
  The correctness of steps~1 and 2 follows from that all proper interval graphs and all the complements of bipartitie permutation graphs are proper circular-arc graphs.
  Step~3 is correct because a solution of $G$ needs to contain a vertex from every forbidden induced subgraph of $G$.
  Since we have passed step~1, there does not exist a solution $V_-$ such that $G - V_-$ is a proper interval graph.  Thus, the resulting graph is connected by Proposition~\ref{lem:connected}, and all its vertices are from a single component $C$ of $G$.  This justifies steps~4 and 5.
  On the other hand, since we have passed step~2, there does not exist a solution $V_-$ such that $\overline G - V_-$ is a bipartite permutation graph.
  If $\overline C$ is bipartite, then any induced subgraph of $\overline C$ is bipartite.
  Thus, the correctness of step~6 follows from Theorem~\ref{thm:co-bipartitie}.
  For the same reason, it suffices to consider each non-bipartite component of $\overline C$ (Proposition~\ref{lem:co-bipartitie}).  Note that each component of $\overline C$ is a reduced graph.  If it is non-bipartite, then its complement is already a proper circular-arc graph (Corollary~\ref{cor:reduced-non-bipartite}).  Thus, it suffices to take the non-bipartite component of $\overline C$ of the maximum order, and remove all the others.  This justifies steps~7 and 8.
 
  We now analyze the running time of the algorithm.
  The first two steps take time $O(6^{k}\cdot (m + n))$~\cite{cao-17-unit-interval-editing} and $9^{k}\cdot n^{O(1)}$~\cite{bozyk-20-bipartite-permutation}, respectively.
  At most $7$ recursive calls are made in step~3, all with parameter value $k - 1$.  Thus, it takes $7^{k}\cdot n^{O(1)}$ time.  There are at most $n$ components in $G$, and the rest (steps 5--8) can be done in $O(m + n)$ time.
  Thus, the algorithm can be done in $9^{k}\cdot n^{O(1)}$ time.

  We now modify the parameterized algorithm to produce an approximation algorithm.   
  We check every set of at most seven vertices of $G$.  If it induces a forbidden induced subgraph, then we remove all its vertices from $G$ and add them to the solution.  We continue this process until no such a set can be found, and let $G'$ be the resulting graph.
  We construct three solutions, and return the smallest of them, with ties broken arbitrarily.  The first is by applying the $6$-approximation algorithm for proper interval vertex deletion~\cite{cao-17-unit-interval-editing} to $G'$, and the second is by applying  the $9$-approximation algorithm for bipartite permutation vertex deletion \cite{bozyk-20-bipartite-permutation} to $\overline{G'}$.  Let them be denoted by $V_-^1$ and $V_-^2$, respectively.  
  For each component, we take the largest non-bipartite component of its complement.   We put all the other vertices into a set.  We take the smallest of these sets as the third solution $V_-^3$.  We return the smallest of the three solutions.

  Let $V^*_-$ be an optimal solution to $G$.  Clearly, $|V(G)\setminus V(G')| \le 7  |V^*_-\setminus V(G')|$.  If $G' - V^*_-$ is a proper interval graph, then $|V_-^1\cap V(G')| \le 6 \mathrm{opt}(G') \le 6  |V^*_-\cap V(G')|$; likewise, if $\overline G - S$ is a bipartite permutation graph, then $|V_-^2\cap V(G')| \le 9  |V^*_-\cap V(G')|$.  If neither is true, then $|V_-^3\cap V(G')| = \mathrm{opt}(G') \le  |V^*_-\cap V(G')|$.  Therefore, the size of the solution returned by the algorithm is never greater than $9|V^*_-|$.
\end{proof}

\appendix
\section*{Appendix: Hardness of edge modifications to proper (Helly) circular-arc graphs}
We note that the NP-completeness of the edge modification problems toward proper (Helly) circular-arc graphs follows from Proposition~\ref{lem:connected}.

\setcounter{section}{1}
\begin{lemma}\label{lem:hardness}
  Both the proper Helly circular-arc edge deletion problem and the proper Helly circular-arc completion problem are NP-complete.
\end{lemma}
\begin{proof}
  Both problems are clearly in NP, and we now show their NP-hardness.
  Let $G$ be a graph, and let $G'$ be the graph obtained from $G$ by adding $|V(G)|^2$ isolated vertices.
  Since $G'$ is not connected, we cannot make it connected by deleting edges.
  By Proposition~\ref{lem:connected}, $(G, k)$ is a yes-instance of  the proper interval edge deletion problem if and only if $(G', k)$ is a yes-instance of the proper Helly circular-arc edge deletion problem.
We now consider the completion problem.  On the one hand, to make $G'$ connected we need to add at least $|V(G)|^2$ edges.  On the other hand, we can always make $G'$ a proper interval graph by adding all possible edges to make each component a clique.  Thus, $(G, k)$ is a yes-instance of  the proper interval completion problem if and only if $(G', k)$ is a yes-instance of the proper Helly circular-arc completion problem.
  The statement then follows from NP-hardness of the corresponding problems toward proper interval graphs \cite{yannakakis-81-minimum-fill-in, goldberg-95-interval-edge-deletion}.  
\end{proof}

Let us remark that the reductions used in the proof of Lemma~\ref{lem:hardness} also apply to edge modification problems toward proper circular-arc graphs.


\end{document}